\documentclass[11pt,letterpaper]{article}

\usepackage[english]{babel}

\usepackage[english]{babel}

\usepackage[letterpaper,margin=1in]{geometry}

\usepackage{amsmath}

\usepackage{mathrsfs}
\usepackage{graphicx}
\usepackage[colorlinks=true, allcolors=blue]{hyperref}
\usepackage[capitalize,noabbrev]{cleveref}
\usepackage{xspace}
\usepackage{array}
\usepackage{microtype}
\usepackage{booktabs}
\usepackage{needspace}
\newcommand{\R}{\mathbb{R}}
\newcommand{\eps}{\varepsilon}
\newcommand{\N}{\mathbb{N}}

\usepackage{amsmath}
\usepackage{amssymb}
\usepackage{mathtools}
\usepackage{amsthm}

\usepackage[capitalize,noabbrev]{cleveref}

\theoremstyle{plain}
\newtheorem{theorem}{Theorem}[section]
\newtheorem{proposition}[theorem]{Proposition}
\newtheorem{lemma}[theorem]{Lemma}
\newtheorem{corollary}[theorem]{Corollary}
\theoremstyle{definition}
\newtheorem{definition}[theorem]{Definition}

\theoremstyle{remark}
\newtheorem{remark}[theorem]{Remark}

\newtheorem{claim}[theorem]{Claim}

\renewcommand{\epsilon}{\varepsilon}

\newcommand{\eat}[1]{}

\DeclareMathOperator*{\E}{\ensuremath{\mathbb{E}}}

\newcommand{\ddim}{\ensuremath{\mathrm{ddim}}}

\newcommand{\calM}{\mathcal{M}}

\newcommand{\calX}{\mathcal{X}}

\newcommand{\cost}{\ensuremath{\mathsf{cost}}}

\usepackage{longtable}

\usepackage{url}

\def\norm#1{\left\| #1 \right\|}

\newcommand{\ip}[2]{\left\langle #1,#2\right\rangle}

\title{Tight Coreset Bounds for $(k,z)$-Clustering in Euclidean and Doubling Metrics}
\author{Junqiu Hu\thanks{
        Email: \texttt{junqiu2025@gmail.com}
    }\\
    Nanjing University
    \and
Lingxiao Huang\thanks{
        Email: \texttt{huanglingxiao@nju.edu.cn}
    }\\
    Nanjing University
    \and
    Jian Li\thanks{
        Email: \texttt{lapordge@gmail.com}
    }\\
    Tsinghua University
    \and
    Baoduo Xu\thanks{ 
        Email: \texttt{bdxu03@gmail.com}}\\
    Nanjing University
}
\date{}
\begin{document}
\maketitle
\begin{abstract}
We study the optimal size of strong coresets for $(k,z)$-clustering and focus on two well-studied classes of metric spaces: Euclidean spaces and doubling metrics.

In Euclidean space, Huang, Li, and Wu (STOC 2024) proved a matching lower bound $\Omega(k \eps^{-z-2})$ when $\eps\gtrsim k^{-1/(z+2)}$, but the upper bound $\widetilde O(k^{(2z+2)/(z+2)}\eps^{-2})$ remained unmatched for smaller $\eps$.

We close this gap with a lower bound matching up to polylogarithmic factors, including for arbitrarily small $\eps$.
Together with existing upper bounds, this determines the optimal worst-case coreset size as
\begin{equation*}
\widetilde\Theta\!\left(\min\left\{k^{\frac{2z+2}{z+2}}\eps^{-2},\,k\eps^{-z-2}\right\}\right)
\end{equation*}
for every fixed $z\ge1$, all $k\ge1$, and $0<\eps<1$.
We replace each offset in the previous hard instance by a group of distinct offsets and construct a fixed family of queries, each using $k$ centers.
Four-wise independent signs yield a discrepancy guarantee: for every sufficiently small subset and every choice of nonnegative weights, some query has a large relative cost error.

For metrics of doubling dimension $D$, the optimal coreset size was known up to logarithmic factors for $1\le z\le2$ when $D\gtrsim\log(k/\eps)$~\cite{CohenAddadLSS22,cohenaddad2021newframework}, but the bounds for $z>2$ did not match.

We obtain a new $\widetilde O(k^{2-2/z}D\eps^{-2})$ upper bound and a new lower bound that, together with the known upper bounds, determine the optimal size
\begin{equation*}
\widetilde\Theta\!\left(kD\min\{k^{1-2/z}\eps^{-2},\,\eps^{-z}\}\right)
\end{equation*}
for fixed $z>2$, $0<\eps\le\eps_z$, and $D\ge C_z\log(2/\eps)$, where $C_z>0$ and $\eps_z\in(0,1)$ depend only on $z$.

The lower bound exploits the greater freedom to prescribe distances in doubling metrics than in Euclidean space: a center can encode $\Theta(D)$ independently chosen signs at a distance scale independent of $D$, without diluting the relative cost discrepancy.

For the upper bound, we introduce a new coreset construction based on independent, unbiased two-point summaries, each exactly preserving its cluster’s mass and anchor cost. We combine this construction with iterative size reduction to obtain the tight size bound.
\end{abstract}

\newpage
\tableofcontents
\newpage
\pagenumbering{arabic}

\section{Introduction}

We study the optimal coreset size for $(k,z)$-clustering, where $z\ge1$ is a fixed constant.\footnote{Throughout, asymptotic notation suppresses constants depending only on the fixed exponent $z$. Unless stated otherwise, tildes additionally suppress factors polynomial in $\log(2+k)$ and $\log(2+1/\eps)$.}

\paragraph{$(k,z)$-clustering.}
Given a metric space $(\calX,d)$, a finite dataset $X\subseteq\calX$, and an integer $k\ge1$, the goal is to find a set $C\subseteq\calX$ of $k$ centers minimizing
\begin{equation*}
\cost_z(X,C):=\sum_{p\in X}d^z(p,C),\qquad d(p,C):=\min_{c\in C}d(p,c).
\end{equation*}
The cases $z=1$ and $z=2$ are the classical $k$-median and $k$-means problems, respectively.
These clustering objectives have been extensively studied in data analysis, machine learning, and approximation algorithms~\cite{Jain2010DataClustering,DBLP:conf/stoc/Huang0024}.

Euclidean space is the most extensively studied geometric setting, with a long line of approximation algorithms for these objectives~\cite{cohenAddadEtAl2022nested,charikarEtAl2026euclidean}.
Another central setting is metrics of bounded \emph{doubling dimension}: a metric has doubling dimension at most $D$ if every ball can be covered by at most $2^D$ balls of half its radius.
This covering property captures low-dimensional structure without requiring a Euclidean embedding; for constant $D$, it supports near-linear-time approximation schemes for clustering~\cite{cohenAddadFeldmannSaulpic2019doubling}.
More recent parameterized approximation schemes cover both Euclidean and doubling metrics~\cite{abbasiEtAl2023normClustering}.
Other well-studied settings include shortest-path metrics of planar and bounded-treewidth graphs~\cite{abbasiEtAl2023normClustering}.

\paragraph{Coresets.}
For large datasets, coresets provide a way to reduce the input size while preserving the clustering objective~\cite{HM04,FL11}.

\begin{definition}[Coreset]
For $0<\eps<1$, a strong $\eps$-coreset for $(k,z)$-clustering of $X$ is a subset $S\subseteq X$ with weights $w:S\to\R_{\ge0}$ such that
\begin{equation*}
\bigl|\cost_z(S,C)-\cost_z(X,C)\bigr|\le\eps\cost_z(X,C)
\quad\text{for every }C\subseteq\calX\text{ with }|C|=k,
\end{equation*}
where $\cost_z(S,C):=\sum_{p\in S}w(p)d^z(p,C)$.
\end{definition}

\noindent
A clustering algorithm can run on a coreset instead of the full dataset while retaining nearly the same approximation guarantee.
An efficiently constructed small coreset separates the cost of reading the data from the more expensive task of optimizing the centers.
This has motivated extensive work on coresets in Euclidean spaces~\cite{HM04,FL11,FeldmanSchmidtSohler2020,huangVishnoi2020coresets}, doubling metrics~\cite{huangJiangLiWu2018doubling,cohenaddad2021newframework}, and graph metrics~\cite{braverman2021excludedminor,cohenAddad2025tightVC}.

\paragraph{Coreset sizes.}
Determining the optimal coreset size in Euclidean space is a central problem in the area; Cohen-Addad, Saulpic, and Schwiegelshohn~\cite{cohenaddad2021powermeans} described improving the Euclidean bounds as arguably ``the most important open problem in coresets.''
Following a sequence of dimension-independent constructions, Cohen-Addad et al.~\cite{CohenAddadLSS22} proved an $\Omega(k\eps^{-2})$ lower bound, which became a natural target for upper bounds.
Subsequent upper bounds of $\widetilde O(k^{4/3}\eps^{-2})$ for $k$-median and $\widetilde O(k^{3/2}\eps^{-2})$ for $k$-means still left open whether a superlinear dependence on $k$ was necessary~\cite{CohenAddadLSSO22}.
Huang, Li, and Wu~\cite{DBLP:conf/stoc/Huang0024} showed that such a superlinear dependence is unavoidable in part of the parameter range: their $\Omega(k\eps^{-z-2})$ lower bound matches the known upper bound for $\eps\gtrsim k^{-1/(z+2)}$.
They also obtained the upper bound $\widetilde O(k^{(2z+2)/(z+2)}\eps^{-2})$ for every fixed $z\ge1$.
The remaining gap lies in the small-error regime $0<\eps<k^{-1/(z+2)}$, where the known lower bounds are $\Omega(k^{(2z+3)/(z+2)}\eps^{-1})$ for $\eps\gtrsim k^{-(z+1)/(z+2)}$ and $\Omega(k\eps^{-2})$ for smaller $\eps$.
These lower bounds do not match the upper bound throughout the small-error regime, leaving the following question.
\begin{quote}
\noindent\textbf{Problem 1.} \emph{For fixed $z\ge1$, what is the optimal worst-case coreset size for Euclidean $(k,z)$-clustering in the remaining regime $0<\eps<k^{-1/(z+2)}$?}
\end{quote}

Doubling metrics provide a second important setting in which to seek tight coreset bounds.
Beginning with the first efficient input-size-independent construction of Huang et al.~\cite{huangJiangLiWu2018doubling}, subsequent work sharpened the dependence on $k$, $\eps$, and $D$.
The framework of Cohen-Addad, Saulpic, and Schwiegelshohn~\cite{cohenaddad2021newframework} yields size $\widetilde O(kD\min\{k\eps^{-2},\eps^{-\max\{z,2\}}\})$.
The $\Omega(kD\eps^{-2})$ lower bound of Cohen-Addad et al.~\cite{CohenAddadLSS22} holds for $D\gtrsim\log(k/\eps)$, matching this upper bound up to logarithmic factors for $1\le z\le2$ in that dimension range.

For $z>2$, however, the dependence on $k$ and $\eps$ remained unresolved.

Huang, Li, and Wu~\cite{DBLP:conf/stoc/Huang0024} proved a lower bound of $\Omega(kD\eps^{-z}/\log k)$ for $\eps\gtrsim k^{-1/(z+2)}$, matching the corresponding upper bound up to logarithmic factors, and an additional lower bound of $\Omega(k^{(2z+1)/(z+2)}D\eps^{-1})$ for $k^{-(z+1)/(z+2)}\lesssim\eps\lesssim k^{-1/(z+2)}$. In both constructions, $D$ denotes the resulting ambient doubling dimension.
For smaller $\eps$, the bounds remained unmatched, raising the following problem:
\begin{quote}
\noindent\textbf{Problem 2.} \emph{For fixed $z>2$, what is the optimal dependence of the worst-case coreset size on $k$, $\eps$, and the doubling dimension $D$ for $0<\eps<k^{-1/(z+2)}$?}
\end{quote}

\Needspace{8\baselineskip}
\subsection{Our Contributions}

We establish matching coreset-size bounds, up to polylogarithmic factors, in Euclidean spaces and doubling metrics.
For Euclidean spaces, our contribution is a lower bound that closes the remaining small-error gap.
For doubling metrics, we improve both upper and lower bounds for $z>2$, resolving the remaining gap when $D\gtrsim\log(2/\eps)$ and $\eps$ is sufficiently small.
Table~\ref{tab:coreset-bounds} summarizes the bounds.

\begin{table}[htbp]
\centering
\small
\setlength{\tabcolsep}{4pt}
\renewcommand{\arraystretch}{1.2}
\newcommand{\coresetcell}[2]{\makebox[#1][c]{\begin{tabular}[c]{@{}c@{}}#2\end{tabular}}}
\resizebox{\textwidth}{!}{\begin{tabular}{|c|@{}c@{}|}
\hline
\multicolumn{2}{|@{}c@{}|}{\begin{tabular}{c|c|c}
\coresetcell{\dimexpr 8.8em+2\tabcolsep+\arrayrulewidth\relax}{\textbf{Metric space}} &
\coresetcell{18em}{\textbf{Previous results}} &
\coresetcell{17em}{\textbf{This work}}
\end{tabular}} \\
\hline
\coresetcell{5em}{Euclidean\\$z\ge1$} &
\begin{tabular}{c|c|c}
\coresetcell{3.8em}{Upper} &
\coresetcell{18em}{
$\min\{k^{\frac{2z+2}{z+2}}\eps^{-2},\,k\eps^{-z-2}\}$\\
\cite{CohenAddadLSS22,DBLP:conf/stoc/Huang0024}
} &
\coresetcell{17em}{---} \\
\hline
\coresetcell{3.8em}{Lower} &
\coresetcell{18em}{
$k\eps^{-2}$~\cite{CohenAddadLSS22}\\
$\min\{k^{\frac{2z+3}{z+2}}\eps^{-1},\,k\eps^{-z-2}\}$~
\cite{DBLP:conf/stoc/Huang0024}
} &
\coresetcell{17em}{
$\min\{k^{\frac{2z+2}{z+2}}\eps^{-2},\,k\eps^{-z-2}\}$ 
}
\end{tabular} \\
\hline
\coresetcell{5em}{Doubling\\$z>2$} &
\begin{tabular}{c|c|c}
\coresetcell{3.8em}{Upper} &
\coresetcell{18em}{
$kD\min\{k\eps^{-2},\,\eps^{-z}\}$ ~
\cite{cohenaddad2021newframework}
} &
\coresetcell{17em}{New upper branch:\\$k^{2-2/z}D\eps^{-2}$} \\
\hline
\coresetcell{3.8em}{Lower} &
\coresetcell{18em}{
$kD\eps^{-2}$~\cite{CohenAddadLSS22}\\
$kD\min\{k^{\frac{z-1}{z+2}}\eps^{-1},\,\eps^{-z}\}$ ~
\cite{DBLP:conf/stoc/Huang0024}
} &
\coresetcell{17em}{
    $kD\min\{k^{1-2/z}\eps^{-2},\,\eps^{-z}\}$
} \\
\end{tabular} \\
\hline
\end{tabular}}
\caption{Coreset-size bounds up to polylogarithmic factors and $z$-dependent constants. The new doubling upper branch holds for all $D\ge1$ and $0<\eps<1$; combined with the previous $kD\eps^{-z}$ branch, it matches our lower bound for $0<\eps\le\eps_z$ and $D\ge C_z\log(2/\eps)$. Other restrictions are as in the cited results and Corollary~\ref{thm:intro-euclidean}.}
\label{tab:coreset-bounds}
\end{table}

We strengthen the construction of Huang, Li, and Wu~\cite{DBLP:conf/stoc/Huang0024} to obtain the following lower bound, which closes the gap throughout $0<\eps<k^{-1/(z+2)}$.

\begin{theorem}[Euclidean coreset lower bound]\label{thm:intro-small-error}
For every fixed real $z\ge1$, every integer $k\ge1$, and every $0<\eps<1$, there exist an integer $d\ge1$ and a finite dataset $X\subseteq\R^d$ such that every $\eps$-coreset for $(k,z)$-clustering of $X$ has size
\begin{equation*}
\widetilde\Omega\!\left(\min\left\{k^{\frac{2z+2}{z+2}}\eps^{-2},\,k\eps^{-z-2}\right\}\right).
\end{equation*}
\end{theorem}

\noindent
Combining this theorem with the known upper bounds~\cite{CohenAddadLSS22,DBLP:conf/stoc/Huang0024} gives the optimal size.
\begin{corollary}[Optimal Euclidean coreset size]\label{thm:intro-euclidean}
For every fixed real $z\ge1$, the lower bound in Theorem~\ref{thm:intro-small-error} is tight up to polylogarithmic factors for every integer $k\ge1$ and $0<\eps<1$, over datasets of unrestricted dimension.
\end{corollary}

\noindent
In particular, the optimal size is $\widetilde\Theta(k^{4/3}\eps^{-2})$ for $k$-median when $\eps<k^{-1/3}$, and $\widetilde\Theta(k^{3/2}\eps^{-2})$ for $k$-means when $\eps<k^{-1/4}$.
Up to polylogarithmic factors, the additional dependence on $k$ remains necessary throughout this regime, including for arbitrarily small $\eps$.
For example, for $k$-means at $\eps=k^{-1}$, the lower bound increases from $\Omega(k^3)$~\cite{CohenAddadLSS22} to $\widetilde\Omega(k^{7/2})$.

For doubling metrics, we improve both the upper and lower bounds, obtaining the following characterization.

\begin{theorem}[Optimal coreset size in doubling metrics]\label{thm:intro-doubling}
For every fixed real $z\ge2$, there are constants $C_z>0$ and $\eps_z\in(0,1)$ such that, for every integer $k\ge1$, every $0<\eps\le\eps_z$, and every $D\ge C_z\log(2/\eps)$, the optimal worst-case size of a strong $\eps$-coreset for $(k,z)$-clustering, over finite positively weighted inputs in metric spaces of doubling dimension at most $D$, is
\begin{equation*}
\widetilde\Theta\!\left(kD\min\{k^{1-2/z}\eps^{-2},\,\eps^{-z}\}\right).
\end{equation*}
\end{theorem}

\noindent
The new upper bound is $\widetilde O(k^{2-2/z}D\eps^{-2})$, valid for all $D\ge1$ and $0<\eps<1$; combining it with the existing $\widetilde O(kD\eps^{-z})$ bound~\cite{cohenaddad2021newframework} gives the combined upper bound in the theorem.
Our lower bound matches it for every prescribed $D$ under the stated conditions, using an ambient space with doubling dimension between $D/8$ and $D$.

For $z>2$, the new upper branch is smaller than the previous $k^2D\eps^{-2}$ branch~\cite{cohenaddad2021newframework} by a factor of $k^{2/z}$.
For sufficiently small $\eps$, our lower bound improves on the previous $\Omega(kD\eps^{-2})$ bound~\cite{CohenAddadLSS22} by a factor of $k^{1-2/z}$.

The two branches of the optimal-size bound meet at $\eps=k^{-1/z}$.

\subsection{Technical Overview}

Starting from the separated-copy framework of Huang, Li, and Wu~\cite{DBLP:conf/stoc/Huang0024}, we replace individual offsets by distinct points and design a fixed family of queries with $k$ centers. The challenge is to ensure that every sufficiently small subset fails some query under any nonnegative reweighting. Shared centers realize four-wise independent sign tests within the center budget; for general $z$, we calibrate perturbations so that powered-cost differences retain a common coefficient. Doubling metrics allow direct distance assignments, letting a center test more groups without increasing the distance scale. For the upper bound, we average unbiased two-point summaries that exactly preserve cluster mass and anchor cost.

\paragraph{Revisiting the previous lower-bound construction.}
At $\eps=\Theta(k^{-1/4})$, Huang, Li, and Wu~\cite{DBLP:conf/stoc/Huang0024} construct $\Theta(k^2)$ points in $\Theta(\sqrt k)$ far-apart copies $P_b$. Within each copy, $j$ indexes spheres and $i$ labels offsets. Suppressing $b$, the local points are
\begin{equation*}
 p_{j,i}=k^{1/4}e_j+u_i,
\end{equation*}
where $e_j$ is a coordinate direction, $k^{1/4}e_j$ is sphere $j$'s center, and $u_i$ is a unit offset. Each copy has $\Theta(\sqrt k)$ spheres and $\Theta(k)$ almost orthogonal offsets. Each offset occurs on $\Theta(\sqrt k)$ spheres, at most once per sphere, giving $\Theta(k^{3/2})$ points per copy.

Suppose a putative coreset $(S,w)$ retains at most half the dataset. Choose an \emph{active copy} $P_1$ with $|S_1|\le |P_1|/2$, where $S_1=S\cap P_1$. The query $C$ uses one center $c_i$ for all occurrences of $u_i$ in $P_1$. The term $-2k^{1/4}\langle e_j,c_i\rangle$ turns a constant change in $c_i$'s $e_j$-coordinate into a $\Theta(k^{1/4})$ cost change on sphere $j$. The other coordinates keep the center norms and $\langle u_i,c_i\rangle$ fixed; the offset design preserves nearest-center assignments. Thus
\begin{equation*}
 d^2(p,C)=
 \begin{cases}
 B,&p\in S_1,\\
 B-\gamma,&p\in P_1\setminus S_1,
 \end{cases}
 \qquad B=\Theta(\sqrt k),\quad\gamma=\Theta(k^{1/4}).
\end{equation*}
The perturbation can be reversed. The other copies are served by their sphere centers, at unit squared cost per point, and the combined query respects the center budget~\cite[Claim~2.2]{DBLP:conf/stoc/Huang0024}.

A distant-center query forces $\sum_{p\in S_1}w(p)=(1\pm O(\eps))|P_1|$~\cite[Claim~2.1]{DBLP:conf/stoc/Huang0024}. Orienting the cost difference to avoid cancellation with the inactive error gives absolute error $\Omega(\gamma|P_1|)-O(\eps B|P_1|)$. Since the active and total costs are both $\Theta(B|P_1|)$, the relative error is at least
\begin{equation*}
 \Omega(\gamma/B)-O(\eps)=\Omega(k^{-1/4})-O(\eps).
\end{equation*}
This proves the $\Omega(k^2)$ bound for a sufficiently small constant in $\eps=\Theta(k^{-1/4})$: one center controls an offset's occurrences on many spheres.

\paragraph{Failed attempts.}
Below $\eps=k^{-1/4}$, the target polynomial size exceeds $k^2$. Repeating each point $q\asymp k^{-1/2}\eps^{-2}$ times gives a multiset $X_{\mathrm{rep}}$ of size $\Theta(k^2q)$ and suggests the required error scale. Draw $N=|X_{\mathrm{rep}}|$ uniform samples with replacement. Each location's count has mean $q$ and standard deviation $\Theta(\sqrt q)$. An active copy has $\Theta(k^{3/2})$ locations with cost differences of order $k^{1/4}$, while the total cost is of order $k^2q$. This suggests
\begin{equation*}
 \frac{k^{1/4}\cdot k^{3/2}\sqrt{q}}{k^2q}
 =\frac{k^{-1/4}}{\sqrt q}
 \asymp\sqrt{\frac{k^{3/2}}{N}}.
\end{equation*}
Equating this expression to $\eps$ gives $N\asymp k^{3/2}\eps^{-2}$. This calculation only suggests the scale: one representative of weight $q$ at each location preserves every query cost exactly. Sampling fluctuations therefore do not establish a lower bound against arbitrary reweighting.

To prevent identical points from collapsing into one representative, a second attempt replaces $u_i$ by distinct offsets $u_{i,s}$:
\begin{equation*}
 p_{j,i,s}=k^{1/4}e_j+u_{i,s},\qquad s\in[q].
\end{equation*}
One center per replacement requires $\Theta(kq)$ active centers, exceeding $k$ as $q$ grows. If instead the offsets with $s\in G\subseteq[q]$ share $c_{i,G}$ across spheres, their distance to that center contains the common term
\begin{equation*}
 -2k^{1/4}\langle e_j,c_{i,G}\rangle
 \qquad\text{for every }s\in G.
\end{equation*}
Changing $\langle e_j,c_{i,G}\rangle$ cannot reproduce the costs $B$ and $B-\gamma$ for retained and omitted replacements in the same group while the other distance terms remain fixed. A common cost change on this group changes the coreset error by that amount times $\sum_{s\in G}(w(p_{j,i,s})-1)$, with $w=0$ on omitted points. We therefore need a family of groupings, each using few centers, such that for every reweighting, at least one grouping detects the aggregate discrepancy. The following properties make this requirement precise.

\paragraph{Identifying the properties needed from the new offsets.}
We keep equal-sized, far-apart copies. For any $(S,w)$ retaining at most half of the new dataset $X$, choose an active copy $P_1$ with at least half its points omitted. We require:
\begin{enumerate}
 \item[P1.] \emph{Distinct replacements.} Each label's $q$ replacements on a sphere are distinct, preventing compression of repeated points into a single representative.
 \item[P2.] \emph{Few centers per grouping.} Each label admits partitions into polylogarithmically many groups, with one center per group shared across spheres.
 \item[P3.] \emph{Balanced costs.} Every query satisfies $\cost_2(P_1,C)\asymp\cost_2(X,C)\asymp |X|$, so inactive costs do not dilute the active error.
 \item[P4.] \emph{Error under arbitrary reweighting.} Some query has absolute coreset error $\widetilde\Omega(k^{1/4}|P_1|/\sqrt q)$, assuming only that at least half of $P_1$ is omitted, not half of each group.
\end{enumerate}
The dataset and query family are fixed before $(S,w)$; only the active copy and witnessing query depend on it. With $\Theta(\sqrt k)$ copies, P3 turns P4 into the relative-error scale suggested by repetition.

\paragraph{Constructing the dataset and sharing centers.}
Keep $b$ for copies, $j$ for spheres, and $i$ for labels; $s$ indexes a label's replacements. Take $n$ far-apart copies, each with $n$ spheres and $m$ labels. Now choose the $u_i$ to be orthonormal, and replace each label's single offset by $\{u_i+x_s:s\in[q]\}$. The local configuration is
\begin{equation*}
 P=\{p_{j,i,s}=\sqrt n\,e_j+u_i+x_s:j\in[n],\ i\in[m],\ s\in[q]\}.
\end{equation*}
The $e_j$, the $u_i$, and the distinct unit vectors $x_s$ lie in three mutually orthogonal subspaces; the $x_s$ need not be mutually orthogonal. Every offset now occurs on every sphere. All triples are distinct, each at distance $\sqrt2$ from its sphere center $\sqrt n\,e_j$, so $|P|=nmq$ and $|X|=n^2mq$, giving P1.

Fix $(S,w)$ and an active copy $P_1$ as above; use local coordinates for its points and query centers. Lemma~\ref{lem:euclidean-code} supplies the $x_s$ and a fixed family $V$ of directions with
\begin{equation*}
 \langle x_s,v\rangle=t_{\ell(s;v)},\qquad
 t_\ell=1-\frac{2\ell}{L},\qquad 0\le\ell\le L,
\end{equation*}
where $\ell(s;v)$ is the level index of $x_s$, and $L=O(\log(q+1))$. One center $c_{i,\ell}$ will serve all points of label $i$ and level $\ell$ across spheres. Inactive copies use their $n(n-1)$ sphere centers, giving the budget
\begin{equation*}
 n(n-1)+m(L+1)\le k.
\end{equation*}
Thus a query needs $m(L+1)$ active centers instead of $mq$. Changing $v$ changes the grouping without increasing this count. A fixed partition with equal costs within each sphere--label--level group could be compressed to one weighted representative per group. To ensure that some grouping detects omissions, Lemma~\ref{lem:euclidean-code} also makes $(-1)^{\ell(s;v)}$ unbiased four-wise independent signs for uniform $v\in V$. The discrepancy bound below guarantees that, for any retained weights, at least one of these groupings detects omissions.

\begin{figure}[t]  
  \centering
    \includegraphics[width=\textwidth]{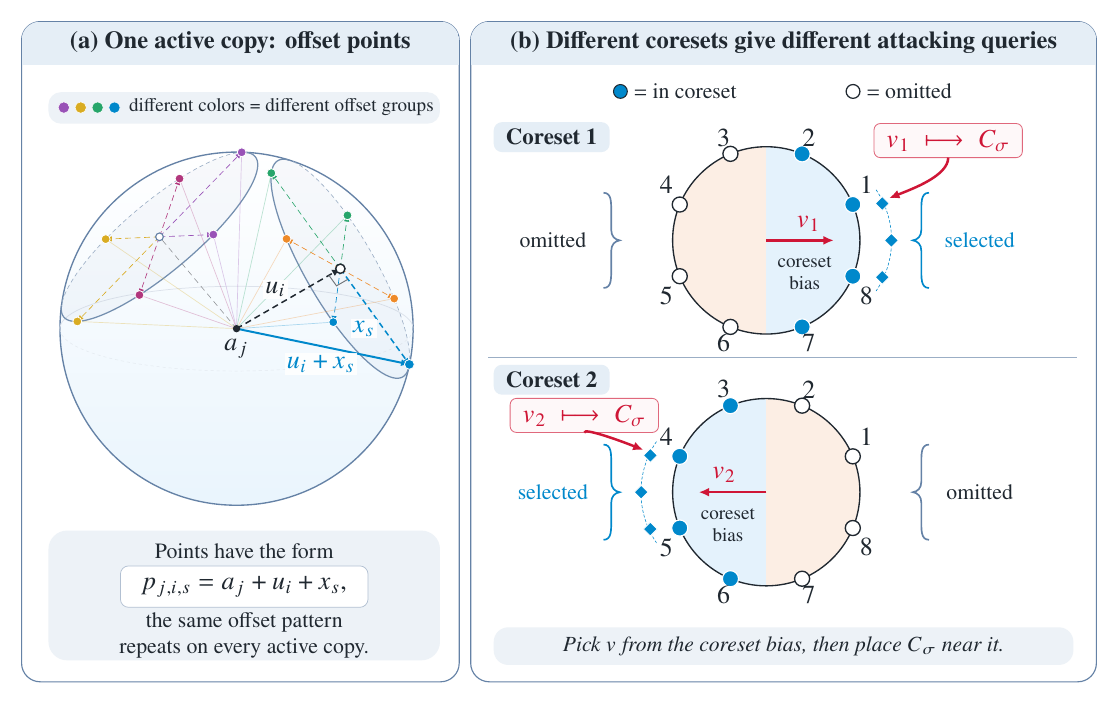}
  \caption{The Euclidean lower-bound construction.
(a) Points have the form $p_{j,i,s}=a_j+u_i+x_s$.
Different colors indicate different offset groups.
The same offset pattern repeats across spheres and copies.
(b) Two candidate coresets select different sets of four points from
the same eight-point group. Filled circles are selected points, each
with weight $2$; hollow circles are omitted points.
For each candidate, we choose a direction $v\in V$ and signs $\sigma$
to obtain a center set $C_\sigma=C(v,\sigma)$ from the fixed query family.
The diamonds illustrate how the centers change with $v$ and $\sigma$.
The query detects a gap in total weighted cost; it need not make
every selected point cheaper than every omitted point.
Only a few centers are shown, and the drawing is not to scale.}
\label{fig:euclidean-instance}
\end{figure}

\paragraph{Realizing level costs with shared centers.}
For signs $\sigma=(\sigma_{ji})\in\{-1,1\}^{n\times m}$, let $C(v,\sigma)$ contain the active centers, the inactive sphere centers, and remote centers padding the count to $k$. The active costs are
\begin{equation}\label{eq:intro-cost-identity}
 d^2(p_{j,i,s},C(v,\sigma))
 =D_{\ell(s;v)}+\frac{\sqrt n}{8L^2}\sigma_{ji}(-1)^{\ell(s;v)},
 \qquad D_\ell=\Theta(n),
\end{equation}
where the $D_\ell$ are the baselines in Eq.~\eqref{eq:pointwise_cost}, independent of $v,\sigma$. The sphere-axis component of $c_{i,\ell}$ is proportional to $-(-1)^\ell\sum_j\sigma_{ji}e_j$, encoding one sign per sphere without changing its norm when signs are reversed. The label component separates labels. For a point at level $\ell$, the projection component and an orthogonal norm adjustment make a competing level $\ell'$ incur the penalty $\sqrt n(t_{\ell'}-t_\ell)^2$. With grid spacing $2/L$, this penalty dominates any change in the signed term. Claim~\ref{clm:euclidean-nearest-centers} verifies that $c_{i,\ell(s;v)}$ remains nearest for every $\sigma$, giving P2. Inactive points have squared cost $2$, so the active and total costs are both $\Theta(n|P_1|)=\Theta(|X|)$, giving P3.

\paragraph{Defeating arbitrary reweighting with opposite queries.}
Extend $w$ by zero to omitted points, and write $\Delta(C)=\cost_2(S,C)-\cost_2(X,C)$. The queries $C(v,\sigma)$ and $C(v,-\sigma)$ have identical baselines and inactive costs. Subtracting their signed errors gives
\begin{equation*}
 \Delta(C(v,\sigma))-\Delta(C(v,-\sigma))
 =\frac{\sqrt n}{4L^2}\sum_{j,i}\sigma_{ji}
 \sum_s(w(p_{j,i,s})-1)(-1)^{\ell(s;v)}.
\end{equation*}
The common baselines and all inactive contributions cancel exactly, without an upper bound on the retained weights. For each fixed $v$, choose $\sigma_{ji}$ to match the sign of its inner sum. Lemma~\ref{lem:euclidean-code} then supplies one common direction with
\begin{align*}
 \max_{v\in V}\sum_{j,i}\left|\sum_s(w(p_{j,i,s})-1)(-1)^{\ell(s;v)}\right|
\ge\frac1{\sqrt3}\sum_{j,i}\left(\sum_s(w(p_{j,i,s})-1)^2\right)^{1/2}\ge\frac{|P_1|}{2\sqrt{3q}}.
\end{align*}
For the last inequality, each omitted point contributes $1$ to a squared discrepancy. Each $q$-point group's norm is at least its number of omissions divided by $\sqrt q$, and $P_1$ omits at least half its points. The common direction makes the aggregate discrepancy large; individual groups need not attain their maxima at that direction.

The opposite queries therefore have signed errors differing by $\Omega(\sqrt n\,|P_1|/(L^2\sqrt q))$. At least one query has absolute error at least half this difference. Its true cost is $\Theta(n|P_1|)$, giving relative error $\Omega(1/(\sqrt{nq}\,L^2))$. For $n\asymp\sqrt k$, this proves P4 and recovers the scale suggested by repetition, now for every reweighting. Proposition~\ref{prop:query-discrepancy} formalizes this argument.

\paragraph{Choosing parameters to obtain the lower bound.}
In the main parameter range, the center budget permits $n\asymp\sqrt k$ and $m\asymp k/L$. The relative-error bound then allows $q=\widetilde\Theta(k^{-1/2}\eps^{-2})$, giving $|X|=n^2mq=\widetilde\Theta(k^{3/2}\eps^{-2})$. Every coreset must retain more than half of these points. Section~\ref{sec:center-set} makes these parameter choices precise and handles the integer boundary cases.

\paragraph{Extension to a general power $z$.}
For $z\ne2$, raising Eq.~\eqref{eq:intro-cost-identity} to $z/2$ produces a level-dependent difference between opposite queries. Since $\ell(s;v)$ changes with $v$, the coefficient of $(-1)^{\ell(s;v)}$ in the error sum would vary with the direction. The discrepancy lemma requires coefficients that do not depend on the direction. Constant-factor estimates do not suffice for this signed sum.

Choose a sphere-axis displacement $\delta_\ell$ for each level so that
\begin{equation*}
 (D_\ell+2\sqrt n\,\delta_\ell)^{z/2}
 -(D_\ell-2\sqrt n\,\delta_\ell)^{z/2}
 =\frac{n^{(z-1)/2}}{32L^2}.
\end{equation*}
Readjusting the orthogonal center component preserves $D_\ell$. Monotonicity and the mean value theorem give the required displacement without changing the nearest-center assignments; see Section~\ref{sec:euclidean-extensions}. The powered-cost difference is now a fixed coefficient times $\sigma_{ji}(-1)^{\ell(s;v)}$, so the same discrepancy argument applies.

An active point costs $\Theta(n^{z/2})$. Using $\lceil n^{z/2}\rceil$ copies, each with $n$ spheres, preserves the cost balance and uses $O(n^{(z+2)/2})$ inactive centers. The relative-error scale remains $1/(\sqrt{nq}\,L^2)$. Hence
taking $n\asymp k^{2/(z+2)}$, $L\asymp \log(2+k^{-1/(z+2)}\eps^{-1})$, $m\asymp k/L$, and $q\asymp\eps^{-2}/(nL^4)$ gives
\begin{equation*}
 |X|=\lceil n^{z/2}\rceil nmq
 =\widetilde\Theta(k^{(2z+2)/(z+2)}\eps^{-2}).
\end{equation*}
Section~\ref{sec:euclidean-extensions} completes the parameter choices and the larger-error case of Theorem~\ref{thm:intro-small-error}.

\paragraph{Lower bounds in doubling metrics (Theorem~\ref{thm:intro-doubling}).}
Doubling metrics allow more flexible distance assignments than Euclidean
space. This lets one center control more groups without increasing
its distance scale.
Take $B$ separated copies, each with $nmq$ distinct points $p_{b,j,i,s}$
at unit distance from an anchor. Each label $i\in[m]$ has $n$ groups
of $q$ points, indexed by $j$.
An active copy uses one center per label, with
\begin{equation*}
 d(p_{b,j,i,s},c_{b,i})=R-\sigma_{bi}^{(j)}\chi_s,\qquad R\ge4.
\end{equation*}
Here $\sigma_{bi}\in\{-1,1\}^n$, and $\chi\in\mathcal H$ ranges over a
fixed polynomial-size family supporting unbiased, four-wise independent signs.

Other labels' centers have distance $R+1$, so $c_{b,i}$ serves all
its label's groups without projection levels. Inactive copies use
their anchors.

The dataset and candidate centers are fixed before
seeing $(S,w)$.

For a candidate $(S,w)$ with $|S|\le|X|/2$, activate the $g$ copies with most omissions.
They omit at least half their points in total, not necessarily in
each group. Subtract the signed errors of $C^\pm=C(\pm\sigma,\chi)$.
Inactive contributions cancel; each active $p=p_{b,j,i,s}$ satisfies
\begin{equation*}
 d^z(p,C^-)-d^z(p,C^+)
 =[(R+1)^z-(R-1)^z]\sigma_{bi}^{(j)}\chi_s.
\end{equation*}
The coefficient is constant across groups, so no calibration is needed.
For fixed $\chi$, choose $\sigma$ so that group contributions add
rather than cancel.
Lemma~\ref{lem:doubling-signs} gives one common $\chi$ with error gap
$\Omega(gR^{z-1}nm\sqrt q)$ for arbitrary nonnegative retained weights.
One query has absolute error at least half this gap and true cost
$O((gR^z+B)nmq)$. Balancing $B\asymp gR^z$ therefore gives relative
error $\Omega(1/(R\sqrt q))$, rather than the Euclidean scale
$1/\sqrt{nq}$ up to logarithmic factors. Increasing $n$ adds points
without weakening this error.

Counting all $nmq$ inputs and $m2^n|\mathcal H|$ candidates per copy
gives dimension $O(n+\log(m+1)+\log(q+1))$, since
$|\mathcal H|=\operatorname{poly}(q+1)$.
A radius-two ball contains $2^n$ candidates requiring separate
radius-one balls, giving dimension $\Omega(n)$.
The choices below satisfy $m\lesssim\eps^{-z}$ and $q\lesssim\eps^{-2}$.
Thus $D\gtrsim\log(2/\eps)$ allows $n\asymp D$ while keeping
dimension at most $D$.
Joining copies and padding points along a line avoids dependence
on their number.

Choose $q\asymp\eps^{-2}R^{-2}$ with a small constant to make the
relative error exceed $\eps$.
Allocate $B\asymp k$ and $gm\asymp k$ within the center budget
$B+g(m-1)\le k$.
Together with cost balance, this gives $m\asymp R^z$, $g\asymp k/R^z$, and
\begin{equation*}
 |X|=Bnmq\asymp kD\eps^{-2}R^{z-2},\qquad
 R^z\lesssim\min\{k,\eps^{-z}\},
\end{equation*}
where the restrictions follow from $g,q\ge1$.
Maximizing $R$ yields
\begin{equation*}
 |X|\asymp kD\min\{k^{1-2/z}\eps^{-2},\eps^{-z}\}.
\end{equation*}
Every coreset must retain more than half these points.
The first branch uses $g\asymp1$; the second has $R\asymp\eps^{-1}$,
$q\asymp1$, and $m\asymp\eps^{-z}$.
We then need $g\asymp k\eps^z$ active copies to use the center budget.
For $z=2$, both branches give $\Theta(kD\eps^{-2})$.
Section~\ref{sec:doubling-lower-proof} verifies the metric and
handles boundary cases.

\paragraph{Upper bounds in doubling metrics (Theorem~\ref{thm:intro-doubling}).}
For fixed $z\ge2$, partition $X$ around a constant-factor solution $A$
with at most $k$ anchors. Write $M_i=\sum_{p\in P_i}u_p$ and
$V_i=\cost_z(P_i,A)$ for the mass and anchor cost of cluster $P_i$ at $a_i$.
Preserving $M_i$ cancels the anchor term from the error, leaving
$g_C(p)=d^z(p,C)-d^z(a_i,C)$. To control it, sample nonnegative weights
$\widehat u$ preserving $V_i$ as well:
\begin{equation*}
 \sum_{p\in P_i}\widehat u_p=M_i,\qquad
 \sum_{p\in P_i}\widehat u_p d^z(p,A)=V_i,\qquad
 \E\widehat u_p=u_p.
\end{equation*}
The two exact constraints define a polytope whose vertices have
at most two nonzero weights.
Sampling vertices from a convex decomposition of the input weights
gives unbiased two-point summaries (Lemma~\ref{lem:doubling-quadrature}).
We average independent summaries; their coordinates need not be independent.
No ring decomposition is needed.

H\"older's inequality gives
$\sum_{p\in P_i}\widehat u_p d(p,A)\le M_i^{1-1/z}V_i^{1/z}$.
Together with the metric inequalities and anchor-cost preservation,
this yields, on every draw,
\begin{equation*}
 \left|\sum_{p\in P_i}\widehat u_p g_C(p)\right|
 \lesssim V_i+d(a_i,C)^{z-1}M_i^{1-1/z}V_i^{1/z}.
\end{equation*}
Each draw is bounded by cluster totals, not the largest individual residual.

Let $V=\sum_iV_i>0$, $\alpha_i=V_i/V$, and take
$s_i=1+\lceil T\alpha_i^{2/z}\rceil$ summaries in $P_i$.
After averaging, their squared ranges sum to $O(1/s_i)$ times
the preceding bound squared, canceling $V_i^{2/z}$.
Since $2-2/z\ge1$, each remaining sum of powers is bounded by the
power of the sum.
Using $\sum_iM_i d^z(a_i,C)\lesssim\cost_z(X,C)+V$ and
$V\lesssim\cost_z(X,C)$ gives a total of $O(\cost_z(X,C)^2/T)$.
For each fixed query, Hoeffding bounds the probability of relative
error exceeding $\eps$ by $2\exp(-\Omega(\eps^2T))$.
Concavity gives $\sum_i\alpha_i^{2/z}\le k^{1-2/z}$, hence
$|S|\le2\sum_i s_i\lesssim k+Tk^{1-2/z}$.
Zero-cost clusters are exact, including when $V=0$.

To cover all queries, discard residuals where
$d(p,A)<c_z\eps d(a_i,C)$ for small $c_z>0$.
These are $O(\eps)d^z(p,C)$ and controlled directly for input and sample.
For the others, record a smallest covering ball for each center,
together with its scale and the cluster cutoffs. At $p$, only records
at scales at most $d(p,A)$ are needed: they retain the nearest centers
to $p$ and $a_i$. Comparing queries with the same records controls
both nearest distances in both directions.
Lemma~\ref{lem:doubling-covering} gives a cover with logarithmic size
$H=O(k\log(2|X|)+kD\log(2/\eps))$ and pointwise error
$O(\eps)d^z(p,A)$, whose sum is bounded by anchor-cost preservation.
For each vector, fix before sampling a representative with
$\cost_z(X,C)+V$ at most twice that of any query mapped to it.
Its concentration bound then applies to all those queries.
A union bound with $T=O(\eps^{-2}H)$ gives constant success probability.

Finally, iterative size reduction~\cite{braverman2021excludedminor}
gives a weighted coreset of size $\operatorname{poly}(k,D,\eps^{-1})$.
Apply the sampling bound again, splitting the error budget between stages.
Coreset composition gives
\begin{equation*}
 O\!\left(k^{2-2/z}\eps^{-2}[D\log(2/\eps)+\log(2k)]\right)
 =\widetilde O(k^{2-2/z}D\eps^{-2}).
\end{equation*}
Section~\ref{sec:doubling-upper-proof} gives the details.

\subsection{Other Related Work}

\paragraph{Sharper bounds for clustering coresets.}
Bansal et al.~\cite{bansal2024sensitivitysampling} showed that
sensitivity sampling constructs Euclidean $k$-means coresets of
size $\widetilde O(\min\{k^{3/2}\eps^{-2},k\eps^{-4}\})$.
For instances with constant cost stability, the same algorithm
achieves size $\widetilde O(k\eps^{-2})$ without knowing the
stability parameter.
Using VC dimension as a structural parameter, Cohen-Addad et al.
\cite{cohenAddad2025tightVC} obtained improved $k$-median coreset bounds, including
$\widetilde O(k\eps^{-2})$ for shortest-path metrics in planar graphs.

\paragraph{Coresets for constrained clustering.}
Coresets for capacitated and fair clustering, as well as more general
assignment constraints, have been studied extensively~\cite{huangJiangVishnoi2019fair,bandyapadhyayFominSimonov2021fair,braverman2022powerUniformSampling,huangLiLuWu2025constrained,prabhuEtAl2027dualConcentration}.
For Euclidean capacitated $(k,z)$-clustering with fixed $z\ge1$,
Huang et al.~\cite{huangLiLuWu2025constrained} obtained coresets of size
$\widetilde O(k^2\eps^{-2z-2})$.\footnote{For fair clustering,
the corresponding bound has an additional factor $\Gamma$, the number
of distinct group-membership profiles.}

\paragraph{Coresets for robust clustering.}
Coresets for clustering with outliers have also received substantial
attention~\cite{huangJiangLouWu2023robust,huangLiLuWu2025constrained,jiangLou2025robustReductions,fangHuangYang2025robustMedian,huangJiangLiWu2025tightRobust}.
For Euclidean $(k,z)$-clustering with $m$ outliers and fixed $z\ge1$,
Huang et al.~\cite{huangJiangLiWu2025tightRobust} obtained coresets of size
$O(m\eps^{-z})+\widetilde O(\min\{k^{(2z+2)/(z+2)}\eps^{-2},k\eps^{-z-2}\})$,
and an alternative bound of $O(m)+\widetilde O(k\eps^{-2z-2})$.

\section{Optimal Euclidean lower bound}
This section proves Theorem~\ref{thm:intro-small-error}. The first four subsections treat $k$-means with $0<\eps<k^{-1/4}$, and the final subsection extends the construction to general powers $z\ge1$ across both the small- and large-error regimes, $0<\eps<k^{-1/(z+2)}$ and $k^{-1/(z+2)}\le\eps<1$.

\subsection{Construction of the dataset}\label{sec:euclidean-main-z2}

We follow the far-apart-copy framework of~\cite{DBLP:conf/stoc/Huang0024}, in which each copy contains points on spheres centered along coordinate axes. We construct each offset from two orthogonal components: one identifies its label, and the other determines its projection level. Let $n,m,q,L\ge1$ be integer parameters to be chosen later. For now, assume that $2^L>q+\binom q3$, as required by Lemma~\ref{lem:euclidean-code}; we verify this condition when choosing the parameters. The dataset $X\subseteq\R^d$, where $d=n+m+2L+1$, consists of $n$ sufficiently separated translates of the set
\begin{equation}
p_{j,i,s}=\sqrt n\,e_j+u_i+x_s,\qquad P=\{p_{j,i,s}:j\in[n],\ i\in[m],\ s\in[q]\}.
\end{equation}
Here $e_1,\ldots,e_n$ and $u_1,\ldots,u_m$ are orthonormal bases of mutually orthogonal subspaces $\R^n$ and $\R^m$. The distinct unit vectors $x_1,\ldots,x_q$ lie in a third orthogonal subspace $\R^{2L}$; fix them, together with the set $V\subseteq\R^{2L}$, as supplied by Lemma~\ref{lem:euclidean-code} below. We reserve the remaining coordinate direction, with unit vector $f$, for adjusting center norms. Because the three summands of $p_{j,i,s}$ lie in mutually orthogonal subspaces, distinct triples $(j,i,s)$ yield distinct points, giving $|P|=nmq$ and $|X|=n^2mq$. Moreover, since $\norm{p_{j,i,s}-\sqrt n\,e_j}^2=\norm{u_i+x_s}^2=2$, each copy is partitioned in local coordinates into $n$ spheres of radius $\sqrt2$ centered at $\sqrt n\,e_1,\ldots,\sqrt n\,e_n$, with $mq$ points per sphere. Our queries will incur squared cost $\Theta(n)$ per point in one \emph{active copy} and $O(1)$ per point in the remaining $n-1$ copies, keeping both contributions of order $|X|$. The following lemma provides the vectors $x_s$ and a family of directions $V$ that induce a small number of projection levels while certifying a large discrepancy on the active copy.

\begin{lemma}[Discrete projections and a discrepancy bound]\label{lem:euclidean-code}
Let $L,q\ge1$ be integers satisfying $2^L>q+\binom q3$. There exist distinct vectors $x_1,\ldots,x_q\in\R^{2L}$ and a nonempty finite set $V\subseteq\R^{2L}$ with the following properties:
\begin{itemize}
\item Every $x_s$ has norm $1$, and every $v\in V$ has squared norm $2$.
\item For every $s\in[q]$ and $v\in V$,
\begin{equation}\label{eq:inner_products_simple}
\ip{x_s}{v}\in\left\{1-\frac{2\ell}{L}:\ell\in\{0,\ldots,L\}\right\}.
\end{equation}
Define the level index $\ell(s;v)\in\{0,\ldots,L\}$ by $\ip{x_s}{v}=1-2\ell(s;v)/L$.
\item For every integer $t\ge1$ and every collection of fixed coefficient vectors $A^{(1)},\ldots,A^{(t)}\in\R^q$,
\begin{equation}\label{eq:euclidean-sign-test}
\max_{v\in V}\sum_{r=1}^t\left|\sum_{s=1}^q A_s^{(r)}(-1)^{\ell(s;v)}\right| \ge\frac1{\sqrt3}\sum_{r=1}^t\|A^{(r)}\|_2.
\end{equation}
\end{itemize}
The vectors $x_s$ and the set $V$ depend only on $L,q$, not on the coefficient vectors.
\end{lemma}

The proof is deferred to Section~\ref{sec:euclidean-main-lemmas}. Each $v\in V$ groups the indices $s\in[q]$ according to their projection levels $\ell(s;v)$. Having only $L+1$ levels allows the construction in Section~\ref{sec:euclidean-centers} to use just $L+1$ centers for each $i\in[m]$. Varying $v$ allows both the grouping and the parity signs $(-1)^{\ell(s;v)}$ to change; the discrepancy bound guarantees a single direction $v\in V$ common to all sphere--label pairs that produces a large coreset error, as shown in Section~\ref{sec:center-set}.

\subsection{A family of query center sets}\label{sec:euclidean-centers}

Fix a candidate weighted subset $(S,w)$ with $|S|\le|X|/2$. By averaging, at least one of the $n$ copies has at most half of its points in $S$; we fix such a copy $P_1$ as the \emph{active copy} and call the remaining $n-1$ copies \emph{inactive}. Our query family places adjustable centers in $P_1$ and a fixed center at each of the $n(n-1)$ sphere centers of the inactive copies. Working in the local coordinates of $P_1$ (so that points $p_{j,i,s}\in P_1$ and their active centers undergo the same translation), we write $w(p_{j,i,s})$ for the weight assigned by $S$ to $p_{j,i,s}\in P_1$, setting $w(p_{j,i,s})=0$ if $p_{j,i,s}\notin S$.

We call $i\in[m]$ the \emph{label} of the offsets $u_i+x_s$ and their corresponding points on each sphere. For a fixed $v\in V$, the \emph{level} of $p_{j,i,s}$ is $\ell(s;v)$, determined by $\ip{x_s}{v}=1-2\ell(s;v)/L$. The points on sphere $j$ with label $i$ and level $\ell$ form the cell
\begin{equation*}
X_{j,i,\ell}(v)=\{p_{j,i,s}:s\in[q],\ \ell(s;v)=\ell\},\qquad j\in[n],\quad i\in[m],\quad 0\le\ell\le L.
\end{equation*}
Because $\ell(s;v)$ depends only on $s$ and $v$, this partition into cells (some of which may be empty) is identical across all $n$ spheres. A single active center $c_{i,\ell}(v,\sigma)$ will serve the cells $X_{j,i,\ell}(v)$ across all spheres $j\in[n]$, reducing the number of active centers from $nm(L+1)$ to $m(L+1)$ (see Figure~\ref{fig:euclidean-instance}).

Our family $C(v,\sigma)$ is parameterized by $v\in V$ and a sign array $\sigma\in\{-1,1\}^{n\times m}$. For each fixed $v$, the construction must keep the nearest-center assignments independent of $\sigma$. For a point on sphere $j$ with label $i$ and level $\ell$, its squared cost should consist of a sign-independent baseline depending only on $\ell$, plus a perturbation proportional to $\sigma_{ji}(-1)^\ell$ with the same positive coefficient for all points. Let $t_\ell:=1-2\ell/L$ be the $\ell$-th projection value in~\eqref{eq:inner_products_simple}. For each $(i,\ell)\in[m]\times\{0,\ldots,L\}$, define the local center
\begin{equation}
c_{i,\ell}(v,\sigma) =2\sqrt n\,u_i+\sqrt n\,t_\ell v -\frac{(-1)^\ell}{16L^2}\sum_{j=1}^n\sigma_{ji}e_j+\lambda_\ell f,
\end{equation}
where $f$ is a unit vector orthogonal to every data point, so its coefficient adjusts the center norm without changing point--center inner products. Choose
\begin{equation*}
\lambda_\ell =\left(4n+(\sqrt n-2n)t_\ell^2-\frac{n}{256L^4}\right)^{1/2}.
\end{equation*}
The square root is well-defined: $\sqrt n-2n\le0$, $t_\ell^2\le1$, and $L\ge1$ imply that its radicand satisfies
\begin{equation*}
4n+(\sqrt n-2n)t_\ell^2-\frac{n}{256L^4}
\ge\frac{511}{256}n+\sqrt n>0.
\end{equation*}
The sign component has squared norm $\left\|\sum_{j=1}^n\sigma_{ji}e_j\right\|^2
=\sum_{j=1}^n\sigma_{ji}^2=n$.
The remaining norms are $\norm{u_i}=\norm{f}=1$ and $\norm{v}^2=2$, the latter by Lemma~\ref{lem:euclidean-code}. Combining these values with the sign-norm calculation gives
\begin{equation}\label{eq:euclidean-center-norm}
\begin{aligned}
\norm{c_{i,\ell}(v,\sigma)}^2
&=4n+2nt_\ell^2+\frac{n}{256L^4}+\lambda_\ell^2\\
&=8n+\sqrt n\,t_\ell^2.
\end{aligned}
\end{equation}
The center norm is therefore independent of $v$ and $\sigma$. We include a center for every label--level pair, even when its cells are empty. Different labels have different $u_i$-components; for a fixed label, different levels have different $v$-components because the values $t_\ell$ are distinct and $\norm{v}^2=2$. Thus the $m(L+1)$ active centers are distinct.

The full center set contains these $m(L+1)$ active centers and the $n(n-1)$ fixed sphere centers in the inactive copies. The parameter choice in Section~\ref{sec:center-set} will ensure that $m(L+1)+n(n-1)\le k$. We obtain exactly $k$ distinct centers by placing the remaining centers sufficiently far from $X$, using the same additional centers for every query. The following claim verifies the assignments and costs of the resulting center sets.

\begin{claim}[Nearest centers and costs]\label{clm:euclidean-nearest-centers}
For every $v\in V$ and $\sigma\in\{-1,1\}^{n\times m}$, the center set $C(v,\sigma)$ has the following properties:
\begin{itemize}
\item \emph{Nearest-center assignment.} For every $j\in[n]$, $i\in[m]$, and $0\le\ell\le L$, the unique nearest center of each point $p\in X_{j,i,\ell}(v)$ is $c_{i,\ell}(v,\sigma)$.
\item \emph{Pointwise costs.} For every $p\in X_{j,i,\ell}(v)$,
\begin{equation}\label{eq:pointwise_cost}
d^2(p,C(v,\sigma))=D_\ell+\frac{\sqrt n}{8L^2}\sigma_{ji}(-1)^\ell,
\end{equation}
where $D_\ell=9n+2-4\sqrt n-\sqrt n\,t_\ell^2$, and
\begin{equation}\label{eq:pointwise_cost_bound}
5n\le d^2(p,C(v,\sigma))\le9n.
\end{equation}
Every point in an inactive copy has squared cost $2$.
\item \emph{Total cost.} The full dataset satisfies
\begin{equation*}
\cost_2(X,C(v,\sigma))\le11|X|.
\end{equation*}
\end{itemize}
\end{claim}

\begin{proof}
Fix $p=p_{j,i,s}\in X_{j,i,\ell}(v)$, so that $\ip{x_s}{v}=t_\ell$. First,
\begin{equation*}
\norm{p}^2=n+2.
\end{equation*}
For any candidate indices $i'\in[m]$ and $0\le\ell'\le L$, bilinearity and orthogonality give
\begin{equation*}
\ip{p}{c_{i',\ell'}(v,\sigma)}
=2\sqrt n\,\mathbf{1}_{\{i'=i\}}
 +\sqrt n\,t_\ell t_{\ell'}
 -\frac{\sqrt n}{16L^2}(-1)^{\ell'}\sigma_{ji'}.
\end{equation*}
Substituting these quantities and~\eqref{eq:euclidean-center-norm} into the squared-distance identity gives
\begin{align*}
&\norm{p-c_{i',\ell'}(v,\sigma)}^2\\
&\quad=9n+2+\sqrt n\,t_{\ell'}^2
       -4\sqrt n\,\mathbf{1}_{\{i'=i\}}
       -2\sqrt n\,t_\ell t_{\ell'}
       +\frac{\sqrt n}{8L^2}(-1)^{\ell'}\sigma_{ji'}\\
&\quad=D_\ell
       +\underbrace{\sqrt n(t_{\ell'}-t_\ell)^2}_{\text{level penalty}}
       +\underbrace{4\sqrt n\,\mathbf{1}_{\{i'\ne i\}}}_{\text{label penalty}}
       +\underbrace{\frac{\sqrt n}{8L^2}(-1)^{\ell'}\sigma_{ji'}}_{\text{sign perturbation}}.
\end{align*}
Here the square completion uses the two identities
\begin{align*}
\sqrt n\,t_{\ell'}^2-2\sqrt n\,t_\ell t_{\ell'}
&=\sqrt n(t_{\ell'}-t_\ell)^2-\sqrt n\,t_\ell^2,\\
-4\sqrt n\,\mathbf{1}_{\{i'=i\}}
&=-4\sqrt n+4\sqrt n\,\mathbf{1}_{\{i'\ne i\}}.
\end{align*}
The $u_i$- and $v$-components produce the label and level penalties. The norm adjustment supplies the quadratic term needed to complete the square and keeps the common cost of order $n$.

Both penalties vanish at $(i',\ell')=(i,\ell)$. If $\ell'\ne\ell$, then
\begin{equation*}
|t_{\ell'}-t_\ell|=\frac{2|\ell'-\ell|}{L}\ge\frac2L.
\end{equation*}
A wrong level therefore incurs at least $4\sqrt n/L^2$. A wrong label incurs $4\sqrt n\ge4\sqrt n/L^2$, since $L\ge1$. Meanwhile,
\begin{equation*}
\left|(-1)^{\ell'}\sigma_{ji'}-(-1)^\ell\sigma_{ji}\right|\le2.
\end{equation*}
Consequently, whenever $(i',\ell')\ne(i,\ell)$,
\begin{align*}
&\norm{p-c_{i',\ell'}(v,\sigma)}^2-\norm{p-c_{i,\ell}(v,\sigma)}^2\\
&\quad\ge\frac{4\sqrt n}{L^2}-\frac{\sqrt n}{4L^2}
       =\frac{15\sqrt n}{4L^2}>0.
\end{align*}
Thus $c_{i,\ell}(v,\sigma)$ is the unique nearest active center. Substituting $(i',\ell')=(i,\ell)$ removes both penalties and yields~\eqref{eq:pointwise_cost}. The signs $\sigma_{ji}$ will align the error contributions across pairs $(j,i)$, while $(-1)^\ell$ distinguishes even and odd levels. The perturbation therefore tests the distribution of weights across levels, not just their total mass.

To verify the constants in~\eqref{eq:pointwise_cost_bound}, first bound the cost from below:
\begin{align*}
\norm{p-c_{i,\ell}(v,\sigma)}^2
&\ge9n+2-\frac{41}{8}\sqrt n
    \qquad\text{($t_\ell^2\le1$, $L\ge1$)}\\
&\ge\frac{103n-9}{16}\ge5n
    \qquad\text{($2\sqrt n\le n+1$, $n\ge1$)}.
\end{align*}
For the upper estimate,
\begin{align*}
\norm{p-c_{i,\ell}(v,\sigma)}^2
&\le9n+2-\frac{31}{8}\sqrt n
    \qquad\text{($t_\ell^2\ge0$, $L\ge1$)}\\
&\le9n.
\end{align*}
Note for Section~\ref{sec:euclidean-extensions} that both bounds in~\eqref{eq:pointwise_cost_bound} remain valid whenever the sign perturbation has magnitude at most $\sqrt n/(8L^2)$.

We also verify that no center from another copy can be closer. In local coordinates,
\begin{equation*}
\norm{p}^2=n+2\le3n,\qquad
\norm{c_{i,\ell}(v,\sigma)}^2\le8n+\sqrt n\le9n,\qquad
\norm{\sqrt n\,e_j}^2=n,
\end{equation*}
where the middle bound uses~\eqref{eq:euclidean-center-norm} and $n\ge1$. Thus all local points and centers have norm at most $3\sqrt n$. The copies can be fixed in advance with translation vectors separated by at least $20\sqrt n$. The reverse triangle inequality then gives a cross-copy distance of at least $20\sqrt n-3\sqrt n-3\sqrt n=14\sqrt n$, whereas the triangle inequality gives a within-copy distance of at most $6\sqrt n$. The copy separation is fixed uniformly over $v$ and $\sigma$.

In an inactive copy, the squared distance from $p_{j,i,s}$ to sphere center $\sqrt n\,e_{j'}$ is
\begin{equation*}
\norm{p_{j,i,s}-\sqrt n\,e_{j'}}^2
=2n\,\mathbf{1}_{\{j'\ne j\}}+2.
\end{equation*}
Thus its own sphere center is uniquely nearest, with squared cost $2$. Additional centers are placed far enough away to affect none of these assignments. The within-copy and cross-copy comparisons establish the global assignments and pointwise costs. Summing these costs gives
\begin{equation}\label{eq:euclidean-cost-bound}
\begin{aligned}
\cost_2(X,C(v,\sigma))
&\le9n|P_1|+2(n-1)|P_1|
    \qquad\text{(by~\eqref{eq:pointwise_cost_bound})}\\
&\le11n|P_1|=11|X|.
\end{aligned}
\end{equation}
\end{proof}

\subsection{Coreset error and parameter choice for $k$-means}\label{sec:center-set}

We now choose two center sets from the preceding family whose signed errors are far apart, forcing at least one to have a large absolute error. Throughout this subsection, $z=2$.

\begin{proposition}[Detecting omissions with paired queries]\label{prop:query-discrepancy}
For the dataset $X$ constructed above, suppose $n(n-1)+m(L+1)\le k$. Let $(S,w)$ be a nonnegatively weighted subset with $|S|\le|X|/2$, and choose an active copy $P_1$ as in Section~\ref{sec:euclidean-centers}. The associated query family contains a center set $C$ such that
\begin{equation}\label{eq:query-discrepancy-bound}
\frac{|\cost_2(S,C)-\cost_2(X,C)|}{\cost_2(X,C)}
\ge\frac{1}{176\sqrt3\,\sqrt{nq}\,L^2}.
\end{equation}
\end{proposition}

\begin{proof}
Extend $w$ by zero to $X$, and write
\[
\Delta(C):=\cost_2(S,C)-\cost_2(X,C)
=\sum_{p\in X}(w(p)-1)d^2(p,C).
\]

\paragraph{Opposite perturbations and sign maximization.}
For a fixed $v\in V$, reversing $\sigma$ keeps the grouping fixed and cancels the baseline $D_\ell$. By Eq.~\eqref{eq:pointwise_cost},
\[
d^2(p_{j,i,s},C(v,\sigma))-d^2(p_{j,i,s},C(v,-\sigma))
=\frac{\sqrt n}{4L^2}\sigma_{ji}(-1)^{\ell(s;v)}.
\]
Inactive points have the same cost under both queries, so their contributions cancel regardless of their weights. Hence
\begin{align*}
&\Delta(C(v,\sigma))-\Delta(C(v,-\sigma))\\
&\quad=\frac{\sqrt n}{4L^2}\sum_{j,i}\sigma_{ji}
\sum_{s=1}^q(w(p_{j,i,s})-1)(-1)^{\ell(s;v)}.
\end{align*}
For each fixed $v$, choose $\sigma_{ji}$ with the sign of its inner sum, taking either sign when that sum is zero. These choices remove cancellation between different sphere--label pairs. Maximizing also over $v$ gives
\begin{equation}\label{eq:euclidean-error-maximum}
\begin{aligned}
&\max_{v,\sigma}\bigl[\Delta(C(v,\sigma))-\Delta(C(v,-\sigma))\bigr]\\
&\quad=\frac{\sqrt n}{4L^2}\max_{v\in V}\sum_{j,i}
\left|\sum_{s=1}^q(w(p_{j,i,s})-1)(-1)^{\ell(s;v)}\right|.
\end{aligned}
\end{equation}

\paragraph{Choosing $v$ and counting omissions.}
The discrepancies within a pair may still cancel for a particular $v$. Apply Lemma~\ref{lem:euclidean-code} to the fixed coefficient vectors $(w(p_{j,i,s})-1)_{s=1}^q$, indexed by $(j,i)\in[n]\times[m]$, to obtain
\begin{equation}\label{eq:euclidean-averaging}
\begin{aligned}
&\max_{v\in V}\sum_{j,i}\left|\sum_{s=1}^q(w(p_{j,i,s})-1)(-1)^{\ell(s;v)}\right|\\
&\quad\ge\frac1{\sqrt3}\sum_{j,i}\left(\sum_{s=1}^q(w(p_{j,i,s})-1)^2\right)^{1/2}.
\end{aligned}
\end{equation}
The maximum is outside the sum, so one common $v$ serves every pair.

Let $M_{ji}$ count omitted points with sphere index $j$ and label $i$ in $P_1$. Each omitted point has discrepancy $-1$, while every squared discrepancy is nonnegative. Since $M_{ji}\le q$ and $\sum_{j,i}M_{ji}\ge|P_1|/2$,
\begin{equation}\label{eq:euclidean-omissions}
\begin{aligned}
\sum_{j,i}\left(\sum_{s=1}^q(w(p_{j,i,s})-1)^2\right)^{1/2}
&\ge\sum_{j,i}\sqrt{M_{ji}}\\
&\ge\frac{\sum_{j,i}M_{ji}}{\sqrt q}
\ge\frac{|P_1|}{2\sqrt q}.
\end{aligned}
\end{equation}
Choose $v,\sigma$ attaining the maximum in Eq.~\eqref{eq:euclidean-error-maximum}, and set $C^+=C(v,\sigma)$ and $C^-=C(v,-\sigma)$. Combining Eqs.~\eqref{eq:euclidean-averaging} and~\eqref{eq:euclidean-omissions} gives
\[
\Delta(C^+)-\Delta(C^-)
\ge\frac{\sqrt n\,|P_1|}{8\sqrt3\,\sqrt q\,L^2}.
\]
At least one query has absolute error at least half this gap. Its true cost is at most $11|X|=11n|P_1|$ by Eq.~\eqref{eq:euclidean-cost-bound}, proving Eq.~\eqref{eq:query-discrepancy-bound}.
\end{proof}

Consequently, at least one query has relative error at least $2\eps$ whenever
\begin{equation}\label{eq:contradiction}
\eps\sqrt{nq}\,L^2\le c_0,\qquad c_0:=\frac{1}{352\sqrt3}.
\end{equation}
Together with the coding condition and the center budget, this forces every $\eps$-coreset to contain more than $|X|/2$ points.

\paragraph{Parameter choice and contradiction.}
Ignoring rounding and fixed constants, the center budget and error condition suggest
\[
n^2\lesssim k,\qquad mL\lesssim k,\qquad q\lesssim\frac{\eps^{-2}}{nL^4}.
\]
For $0<\eps<k^{-1/4}$, choose
\[
L=\left\lceil8\log_2(2+k^{-1/4}\eps^{-1})\right\rceil,\qquad
q=\max\left\{1,\left\lfloor\frac{k^{-1/2}\eps^{-2}}{L^4}\right\rfloor\right\},
\]
and round $n,m$ down from sufficiently small fixed multiples of $\sqrt k,k/L$, respectively. Since $q\le k^{-1/2}\eps^{-2}$,
\[
2^L\ge(2+k^{-1/4}\eps^{-1})^8
>2(k^{-1/2}\eps^{-2})^3\ge q+\binom q3.
\]
Moreover, $k^{-1/4}\eps^{-1}\ge1$ implies $L^4\lesssim k^{-1/2}\eps^{-2}$, so rounding $q$ up to one loses only an absolute factor and $q\asymp k^{-1/2}\eps^{-2}/L^4$.

Suppose first that $n,m\ge1$. Then $n\asymp\sqrt k$ and $m\asymp k/L$, and the sufficiently small leading constants ensure $n(n-1)+m(L+1)\le k$. With $L$ fixed as above,
\[
\eps\sqrt{nq}\,L^2\lesssim\sqrt{\frac{n}{\sqrt k}}.
\]
Taking the leading constant in $n$ sufficiently small makes this quantity at most $c_0$. Thus every $\eps$-coreset satisfies
\[
|S|>\frac{|X|}{2}=\frac12n^2mq
\asymp\frac{k^{3/2}\eps^{-2}}{L^5}
\asymp\frac{k^{3/2}\eps^{-2}}{[\log(2+k^{-1/4}\eps^{-1})]^5}.
\]
If either rounded count vanishes, then $k=O(1)$ or $k=O(L)$, respectively. In either case,
\[
\frac{k^{3/2}\eps^{-2}}{L^5}
=k\eps^{-2}\frac{\sqrt k}{L^5}
=O(k\eps^{-2}).
\]
For $\eps<1/2$, the known $\Omega(k\eps^{-2})$ lower bound for $k$-means~\cite[Theorem~11]{CohenAddadLSS22} therefore suffices. For $\eps\ge1/2$, take $X$ to consist of $k$ distinct points. Any proper weighted subset has zero cost for a query containing its retained points and padded to $k$ distinct centers outside $X$, while the true cost is positive. Thus all $k$ points are necessary; since $\eps^{-2}\le4$, this gives the required $\Omega(k\eps^{-2})$ bound. This completes the $k$-means lower bound throughout $0<\eps<k^{-1/4}$.

\subsection{Proof of Lemma~\ref{lem:euclidean-code}: projection levels and signs}\label{sec:euclidean-main-lemmas}

We now prove the projection lemma used in the preceding construction and error analysis.

\begin{proof}[Proof of Lemma~\ref{lem:euclidean-code}]
We first construct distinct binary words $b_1,\ldots,b_q\in\{0,1\}^L$ such that no four distinct words sum to zero. Sums of binary words are taken coordinatewise modulo two. Suppose $b_1,\ldots,b_{s-1}$ have been chosen. To choose $b_s$, exclude
\begin{equation*}
\{b_{s'}:1\le s'<s\}
\;\cup\;
\{b_{s_1}+b_{s_2}+b_{s_3}:1\le s_1<s_2<s_3<s\}.
\end{equation*}
The number of excluded words is at most
\begin{equation*}
(s-1)+\binom{s-1}{3}<q+\binom q3<2^L
\qquad\text{($s\le q$ and the lemma's hypothesis)}.
\end{equation*}
The two terms count previous words and previous triples; coincident exclusions only reduce the count. Fewer than all binary words are excluded, so a valid choice remains at every step. The first exclusion enforces distinctness; the second rules out a four-word zero sum, since its last word would equal the sum of three earlier words.

Let $h_{a,\beta}$, for $a\in[L]$ and $\beta\in\{0,1\}$, be an orthonormal basis of $\R^{2L}$, and let $b_s(a)$ denote the $a$th bit of $b_s$. Define
\begin{equation*}
x_s=\frac1{\sqrt L}\sum_{a=1}^Lh_{a,b_s(a)},\qquad s\in[q].
\end{equation*}
These vectors have unit norm:
\begin{equation*}
\norm{x_s}^2=\frac1L\sum_{a=1}^L1=1.
\end{equation*}
If two binary words differ at coordinate $a$, their vectors have different coefficients on $h_{a,0}$ and $h_{a,1}$. Hence the vectors are distinct.

Draw the signs $r_{a,\beta}$ independently and uniformly from $\{-1,1\}$, and set
\begin{equation*}
v=\frac1{\sqrt L}\sum_{a=1}^L\sum_{\beta\in\{0,1\}}r_{a,\beta}h_{a,\beta}.
\end{equation*}
Distinct sign arrays have distinct coordinates in the basis, so the $2^{2L}$ arrays yield $2^{2L}$ vectors; let $V$ be their set. Moreover,
\begin{equation*}
\norm{v}^2=\frac1L\sum_{a=1}^L\sum_{\beta\in\{0,1\}}r_{a,\beta}^2=2.
\end{equation*}
Throughout this proof, $\E_v$ denotes the uniform average over $V$, equivalently expectation over the independent signs. This averaging argument establishes a maximum bound for the fixed family; it does not randomize the dataset. Let $\ell(s;v)$ count the coordinates $a$ with $r_{a,b_s(a)}=-1$. Then
\begin{equation*}
\ip{x_s}{v}=\frac1L\sum_{a=1}^Lr_{a,b_s(a)}
=1-\frac{2\ell(s;v)}L.
\end{equation*}
This proves~\eqref{eq:inner_products_simple}. The same count also gives
\begin{equation}\label{eq:parity}
\prod_{a=1}^Lr_{a,b_s(a)}=(-1)^{\ell(s;v)}.
\end{equation}

For the moment calculations, write $\chi_s(v):=(-1)^{\ell(s;v)}$. Fix a coefficient vector $A=(A_s)_{s=1}^q\in\R^q$ and define
\begin{equation*}
Y(v):=\sum_{s=1}^qA_s\chi_s(v).
\end{equation*}
The second moment identifies the scale $\|A\|_2^2$; a fourth-moment bound ensures that this scale is not due only to rare large values. It suffices to prove that for every fixed subset $T\subseteq[q]$ of size $2$ or $4$,
\begin{equation}\label{eq:euclidean-parity-cancellation}
\E_v\prod_{s\in T}\chi_s(v)=0.
\end{equation}
For each $a$, let $N_a:=\sum_{s\in T}b_s(a)$ be an ordinary integer sum, counting the words in $T$ whose $a$th bit is $1$. The other $|T|-N_a$ words have bit $0$. Hence
\begin{equation*}
\prod_{s\in T}\chi_s(v)
=\prod_{a=1}^Lr_{a,0}^{|T|-N_a}r_{a,1}^{N_a}
\qquad\text{(by~\eqref{eq:parity})}.
\end{equation*}
If $|T|=2$, distinctness gives $\sum_{s\in T}b_s\ne0$ modulo two. If $|T|=4$, the exclusion of four-word zero sums gives the same conclusion. Thus some coordinate $a_*$ has $N_{a_*}$ odd. At that coordinate,
\begin{equation*}
\E\bigl[r_{a_*,1}^{N_{a_*}}\bigr]=\E[r_{a_*,1}]=0,
\end{equation*}
because an odd power of a sign equals the sign itself and the sign is uniform. By independence,
\begin{align*}
\E_v\prod_{s\in T}\chi_s(v)
&=\prod_{a=1}^L
    \left(\E[r_{a,0}^{|T|-N_a}]\,\E[r_{a,1}^{N_a}]\right)
    \qquad\text{(independent signs)}\\
&=0 \qquad\text{(the $a_*,1$ factor vanishes)}.
\end{align*}
This proves~\eqref{eq:euclidean-parity-cancellation}. For a product of one or three distinct $\chi_s$, the total multiplicity at each coordinate is odd, so at least one of $r_{a,0},r_{a,1}$ has odd multiplicity. Its expectation is zero as well. Thus the signs are unbiased and four-wise independent; the following moment calculation only needs the even cases.

Every $\chi_s$ is a sign, so $\chi_s^2=1$. Expanding the second moment gives
\begin{equation*}
\E_vY^2=\sum_{s=1}^qA_s^2=\|A\|_2^2
\qquad\text{(by~\eqref{eq:euclidean-parity-cancellation})}.
\end{equation*}
In the fourth moment, multiplicity patterns $3+1$ and $2+1+1$ reduce to products of two distinct signs; pattern $1+1+1+1$ uses four distinct signs. Their expectations vanish by~\eqref{eq:euclidean-parity-cancellation}. The surviving patterns are $4$ and $2+2$, and the latter has $\binom42=6$ orderings. Therefore,
\begin{align*}
\E_vY^4
&=\sum_sA_s^4+6\sum_{s<s'}A_s^2A_{s'}^2\\
&=3\left(\sum_sA_s^2\right)^2-2\sum_sA_s^4
 \le3\|A\|_2^4.
\end{align*}
If $A=0$, the required average bound is immediate. Otherwise $\E_vY^2>0$, so $\E_vY^4>0$ and Lemma~\ref{lem:holder} applies:
\begin{align*}
\E_v|Y|
&\ge\frac{(\E_vY^2)^{3/2}}{(\E_vY^4)^{1/2}}
    \qquad\text{(Lemma~\ref{lem:holder})}\\
&\ge\frac1{\sqrt3}\|A\|_2
    \qquad\text{(the second and fourth moments)}.
\end{align*}
For the vectors $A^{(1)},\ldots,A^{(t)}$ in the lemma, let $Y_r(v):=\sum_sA_s^{(r)}\chi_s(v)$. Summing the expectation bounds for the coefficient vectors and comparing the average with the maximum yields
\begin{align*}
\max_{v\in V}\sum_{r=1}^t|Y_r(v)|
&\ge\E_v\sum_{r=1}^t|Y_r(v)|
 =\sum_{r=1}^t\E_v|Y_r(v)|\\
&\ge\frac1{\sqrt3}\sum_{r=1}^t\|A^{(r)}\|_2
    \qquad\text{(the average bound for each vector)}.
\end{align*}
Substituting the definitions of $Y_r$ and $\chi_s$ gives~\eqref{eq:euclidean-sign-test}, with one common $v$ for every coefficient vector.
\end{proof}

\begin{lemma}[Moment interpolation]\label{lem:holder}
Every real-valued random variable $Y$ with $0<\E|Y|^4<\infty$ satisfies
\begin{equation*}
\E|Y|\ge\frac{(\E|Y|^2)^{3/2}}{(\E|Y|^4)^{1/2}}.
\end{equation*}
\end{lemma}
\begin{proof}
The first and second absolute moments are finite because $|Y|^r\le1+|Y|^4$ for $r\in\{1,2\}$. Factor $|Y|^2$ so that H\"older's exponents are $3/2$ and $3$:
\begin{align*}
\E|Y|^2
&=\E\bigl[|Y|^{2/3}|Y|^{4/3}\bigr]\\
&\le(\E|Y|)^{2/3}(\E|Y|^4)^{1/3}
    \qquad\text{(H\"older, $3/2$ and $3$)}.
\end{align*}
Raising both sides to the power $3/2$, which preserves the inequality between nonnegative quantities, gives
\begin{equation*}
(\E|Y|^2)^{3/2}\le\E|Y|\,(\E|Y|^4)^{1/2}.
\end{equation*}
Dividing by $(\E|Y|^4)^{1/2}>0$ proves the claim.
\end{proof}

\subsection{General powers and remaining parameter ranges}\label{sec:euclidean-extensions}

\paragraph{General powers $z\ge1$.}
Fix a real $z\ge1$ and $0<\eps<1$. We retain the local dataset and projection family, but use $\lceil n^{z/2}\rceil$ copies instead of $n$. Each copy still contains $nmq$ points, so
\[
|X|=\lceil n^{z/2}\rceil nmq\asymp n^{(z+2)/2}mq.
\]
The active points will have cost $\Theta(n^{z/2})$, so this copy count keeps the active contribution comparable to $|X|$. For a weighted subset $(S,w)$ with $|S|\le|X|/2$, extend $w$ by zero to $X$ and choose an active copy $P_1$ containing at most half its points in $S$, as before. Each inactive copy receives its $n$ fixed sphere centers. Together with the $m(L+1)$ active centers, this uses $n(\lceil n^{z/2}\rceil-1)+m(L+1)$ centers; we verify the budget below and pad to exactly $k$ centers as before.

The remaining change is to the perturbation magnitudes. For $z=2$, opposite perturbations give the same pointwise cost difference at every level. For general $z$, taking the $z/2$-th power introduces a dependence on the baseline $D_\ell$. We compensate by choosing, for every $0\le\ell\le L$, a magnitude $\delta_\ell\in(0,1/(16L^2)]$ satisfying
\begin{equation}\label{eq:euclidean-z-powered-gap}
\left(D_\ell+2\sqrt n\,\delta_\ell\right)^{z/2}
-\left(D_\ell-2\sqrt n\,\delta_\ell\right)^{z/2}
=\frac{n^{(z-1)/2}}{32L^2}.
\end{equation}
Both arguments lie in $[5n,9n]$ throughout the allowed interval, by the estimates in Claim~\ref{clm:euclidean-nearest-centers}. Since $(z/2)u^{z/2-1}\ge n^{z/2-1}/6$ on this interval for every $z\ge1$, the mean value theorem gives
\[
\left(D_\ell+\frac{\sqrt n}{8L^2}\right)^{z/2}
-\left(D_\ell-\frac{\sqrt n}{8L^2}\right)^{z/2}
\ge\frac{n^{(z-1)/2}}{24L^2}.
\]
The gap is continuous and strictly increasing from zero to a value above the target, so the required $\delta_\ell$ exists uniquely. It depends only on $z,n,L,\ell$, not on $v,\sigma,S,$ or $w$.

In the active-center formula of Section~\ref{sec:euclidean-centers}, replace $1/(16L^2)$ by $\delta_\ell$ and set
\[
\lambda_\ell=\left(4n+(\sqrt n-2n)t_\ell^2-n\delta_\ell^2\right)^{1/2}.
\]
The radicand only increases, and the squared center norm remains $8n+\sqrt n\,t_\ell^2$. The label and level penalties in Claim~\ref{clm:euclidean-nearest-centers} are unchanged, while each sign perturbation still has absolute value at most $\sqrt n/(8L^2)$. Thus the same nearest-center comparison and separation between copies apply. An active point $p_{j,i,s}$ at level $\ell=\ell(s;v)$ therefore has cost
\[
d^z(p_{j,i,s},C(v,\sigma))
=\left(D_\ell+2\sqrt n\,\delta_\ell\sigma_{ji}(-1)^\ell\right)^{z/2}
\in[(5n)^{z/2},(9n)^{z/2}],
\]
whereas every inactive point has cost $2^{z/2}$. Consequently,
\begin{equation}\label{eq:euclidean-z-cost}
\begin{aligned}
\cost_z(X,C(v,\sigma))
&\le(9n)^{z/2}|P_1|+2^{z/2}(\lceil n^{z/2}\rceil-1)|P_1|\\
&\lesssim n^{z/2}|P_1|\asymp|X|.
\end{aligned}
\end{equation}

Write $\Delta(C)=\cost_z(S,C)-\cost_z(X,C)$. By Eq.~\eqref{eq:euclidean-z-powered-gap}, opposite sign arrays satisfy
\[
d^z(p_{j,i,s},C(v,\sigma))-d^z(p_{j,i,s},C(v,-\sigma))
=\frac{n^{(z-1)/2}}{32L^2}\sigma_{ji}(-1)^{\ell(s;v)}.
\]
The coefficient is independent of the level and direction, and inactive contributions cancel. The same sign maximization and Eqs.~\eqref{eq:euclidean-averaging} and~\eqref{eq:euclidean-omissions} give a pair $C^+=C(v,\sigma)$ and $C^-=C(v,-\sigma)$ with
\[
\Delta(C^+)-\Delta(C^-)
\ge\frac{n^{(z-1)/2}|P_1|}{64\sqrt3\,\sqrt q\,L^2}.
\]
At least one query has absolute error at least half this gap. Together with Eq.~\eqref{eq:euclidean-z-cost}, this yields
\[
\max_{C\in\{C^+,C^-\}}\frac{|\Delta(C)|}{\cost_z(X,C)}
\gtrsim\frac{1}{\sqrt{nq}\,L^2}.
\]
Thus the relative-error scale is the same as for $k$-means. For a sufficiently small constant $c_z>0$ depending only on $z$, this error is at least $2\eps$ whenever
\begin{equation}\label{eq:euclidean-z-contradiction}
\eps\sqrt{nq}\,L^2\le c_z.
\end{equation}
Subject to the coding condition and center budget, this forces every $\eps$-coreset to contain more than $|X|/2$ points.

\paragraph{Small errors: $0<\eps<k^{-1/(z+2)}$.}
As in the $k$-means case, the center budget and error condition suggest
\[
n^{(z+2)/2}\lesssim k,\qquad mL\lesssim k,\qquad q\lesssim\frac{\eps^{-2}}{nL^4}.
\]
Choose
\[
L=\left\lceil8\log_2(2+k^{-1/(z+2)}\eps^{-1})\right\rceil,\qquad
q=\max\left\{1,\left\lfloor\frac{k^{-2/(z+2)}\eps^{-2}}{L^4}\right\rfloor\right\},
\]
and round $n,m$ down from sufficiently small fixed multiples of $k^{2/(z+2)},k/L$, respectively, allowing the multiplier for $n$ to depend on $z$. The same coding and rounding estimates as in Section~\ref{sec:center-set} give
\[
2^L>q+\binom q3,\qquad q\asymp\frac{k^{-2/(z+2)}\eps^{-2}}{L^4}.
\]
Suppose first that both rounded counts are positive. Then $n\asymp k^{2/(z+2)}$ and $m\asymp k/L$. Taking their leading constants sufficiently small ensures
\[
n(\lceil n^{z/2}\rceil-1)+m(L+1)
\le n^{(z+2)/2}+m(L+1)\le k.
\]
Moreover,
\[
\eps\sqrt{nq}\,L^2\lesssim\sqrt{\frac{n}{k^{2/(z+2)}}}.
\]
With $L$ fixed as above, choosing the leading constant in $n$ sufficiently small makes this quantity at most $c_z$. Hence every $\eps$-coreset satisfies
\[
\begin{aligned}
|S|>\frac{|X|}{2}
&\asymp n^{(z+2)/2}mq
\asymp\frac{k^{(2z+2)/(z+2)}\eps^{-2}}{L^5}\\
&\asymp\frac{k^{(2z+2)/(z+2)}\eps^{-2}}{[\log(2+k^{-1/(z+2)}\eps^{-1})]^5}.
\end{aligned}
\]
If either rounded count vanishes, then $k=O(1)$ or $k=O(L)$, respectively, so $k^{z/(z+2)}/L^5=O(1)$ and the target bound is $O(k\eps^{-2})$. For $\eps<1/2$, the known $\Omega(k\eps^{-2})$ lower bound~\cite[Theorem~11]{CohenAddadLSS22} suffices; for $\eps\ge1/2$, the $k$-point argument in Section~\ref{sec:center-set} gives the required $\Omega(k)$ bound. This completes the small-error case of Theorem~\ref{thm:intro-small-error}.

\paragraph{Larger errors: $k^{-1/(z+2)}\le\eps<1$.}
The same construction also recovers the known larger-error branch~\cite{DBLP:conf/stoc/Huang0024}. A single offset per label suffices: for sufficiently small $\eps$ depending only on $z$, choose
\[
L=q=1,\qquad n\asymp\eps^{-2},\qquad m\asymp k,
\]
rounding $n,m$ down from sufficiently small fixed multiples. These counts are positive, and $\eps^{-z-2}\le k$ allows the leading constants to be chosen so that
\[
n(\lceil n^{z/2}\rceil-1)+2m\le n^{(z+2)/2}+2m\le k,
\qquad \eps\sqrt n\le c_z.
\]
The coding condition is immediate, so every coreset has size $|S|>|X|/2\asymp n^{(z+2)/2}m\asymp k\eps^{-z-2}$. When $\eps$ is bounded below by a positive constant depending only on $z$, this target is $O(k)$ and the $k$-point argument again suffices. This proves the larger-error case and completes Theorem~\ref{thm:intro-small-error}; the known upper bounds then give Corollary~\ref{thm:intro-euclidean}.

\section{Lower bound in doubling metrics}\label{sec:doubling-lower-proof}

We prove the lower-bound part of Theorem~\ref{thm:intro-doubling} through the following lemma: it constructs an unweighted dataset in a finite metric space of prescribed doubling dimension for which every $\eps$-coreset must retain more than half of the input points.

The finite metric space consists of the dataset, all candidate centers, and auxiliary anchor and padding points.

\begin{lemma}[Coreset lower bound in doubling metrics]\label{lem:doubling-lower}
For every fixed real $z\ge2$, there exist constants $C_z>0$ and $\eps_z\in(0,1)$ such that for every integer $k\ge1$, every $\eps\in(0,\eps_z]$, and every $D\ge C_z\log(2/\eps)$, one can construct a finite metric space $(\calX,d)$ satisfying
\[
 D/8\le\ddim(\calX,d)\le D
\]
and an unweighted dataset $X\subseteq\calX$ of size
\begin{equation}\label{eq:doubling-lower}
 |X|\asymp kD
       \min\{k^{1-2/z}\eps^{-2},\,\eps^{-z}\}
\end{equation}
for which every $\eps$-coreset for $(k,z)$-clustering of $X$ in $(\calX,d)$ has size greater than $|X|/2$. 
\end{lemma}

\paragraph{Proof of Lemma~\ref{lem:doubling-lower}.}
Unlike in the Euclidean construction, we can prescribe the signed distances directly. A center will serve several groups of points at distances $R-1$ or $R+1$, with an independently chosen sign for each group. The doubling dimension, rather than the distance scale $R$, controls the number of groups. Comparing queries with opposite sign choices then detects omissions despite arbitrary coreset weights.

\subsection{Construction of the ambient metric}

Let $B,n,m,q\ge1$ be integer parameters and let $R\ge4$, to be chosen below. For each $b\in[B]$, create an anchor $a_b$ and a copy
\[
P_b=\{p_{b,j,i,s}:j\in[n],\ i\in[m],\ s\in[q]\}.
\]
The unweighted dataset is
\[
X=\bigcup_{b=1}^B P_b,\qquad |X|=Bnmq.
\]

We index the candidate centers by a finite family $\mathcal H$ of $\pm1$ vectors. This family has the same unbiased four-wise independence property as the Euclidean parity vectors $\{((-1)^{\ell(s;v)})_{s\in[q]}:v\in V\}$ from Lemma~\ref{lem:euclidean-code}. Its discrepancy guarantee will play the same role in detecting omitted points, but no projection levels are needed.

\begin{lemma}[A finite sign family]\label{lem:doubling-signs}
For every integer $q\ge1$, there is a distribution on a finite set $\mathcal H\subseteq\{-1,1\}^q$ such that a random vector $\chi\in\mathcal H$ has unbiased, four-wise independent coordinates and
\begin{equation*}
 2\le|\mathcal H|\le4(q+2)^6,
 \qquad
 \E_{\chi}\left| \left\langle{A,\chi}\right\rangle \right|
 \ge\frac1{\sqrt3}\left\| A \right\|_2
 \quad(A\in\R^q).
\end{equation*}
The family and its distribution depend only on $q$.
\end{lemma}

The proof is given in Section~\ref{sec:doubling-sign-proof}.

\paragraph{Candidate centers and their distances.}

For every copy $b\in[B]$, label $i\in[m]$, sign vector $\sigma_{bi}=(\sigma_{bi}^{(j)})_{j\in[n]}\in\{-1,1\}^n$, and $\chi\in\mathcal H$, add a distinct candidate center $c_{b,i,\sigma_{bi},\chi}$. Each copy contains $m2^n|\mathcal H|$ candidate centers.

We also include $k$ isolated padding points $a_{B+1},\ldots,a_{B+k}$ so that every query can be padded to $k$ centers.

The full point set is
\[
\begin{aligned}
 \calX={}&X\cup\{a_b:b\in[B+k]\}\\
 &{}\cup\{c_{b,i,\sigma_{bi},\chi}:
 b\in[B],\ i\in[m],\ \sigma_{bi}\in\{-1,1\}^n,\ \chi\in\mathcal H\}.
\end{aligned}
\]
A component consists of one copy $P_b$, its anchor $a_b$, and all its candidate centers; each padding point $a_{B+1},\ldots,a_{B+k}$ forms a separate singleton component.

The distances are specified by the following cases.
\begin{itemize}
\item \emph{Input--anchor distances.} $d(p,a_b)=1$ for every $p\in P_b$.
\item \emph{Input--input distances.} For distinct $p,p'\in P_b$, set $d(p,p')=2$.
\item \emph{Anchor--center distances.} For every candidate center $c$ in copy $b$, set $d(a_b,c)=R$.
\item \emph{Center--center distances.} For distinct candidate centers $c,c'$ in the same copy, set $d(c,c')=2$.
\item \emph{Input--center distances.} Within a copy $b$, set
\[
 d(p_{b,j,i,s},c_{b,i',\sigma_{bi'},\chi})=
 \begin{cases}
 R-\sigma_{bi}^{(j)}\chi_s,&i'=i,\\
 R+1,&i'\ne i.
 \end{cases}
\]
\item \emph{Distances between different components.} Assign the $B$ copies and $k$ singleton components fixed consecutive positions. For points in different components, set
\begin{equation}\label{eq:doubling-line-join}
 d(x,y)=10R|t-t'|
 \quad\text{if their components have positions $t$ and $t'$.}
\end{equation}
\item \emph{Diagonal and symmetry.} Finally, set $d(x,x)=0$ and impose symmetry.
\end{itemize}
These cases specify every distance, including those involving padding points. The line arrangement keeps the dimension independent of the number of components; equal cross-distances would not.

Thus all $nmq$ input points in $P_b$ lie in the unit ball centered at $a_b$. For each label $i$, the index $j$ specifies one of $n$ groups, each containing $q$ points indexed by $s$. A single selected center of label $i$ will serve all $n$ groups, with an independent sign choice for each group.

\subsection{The metric and its doubling dimension}\label{sec:doubling-dim-of-the-metric}

We prove that the points and distances specified above form a finite metric space $(\mathcal{X},d)$, and then bound its doubling dimension.

We will choose the parameters to satisfy
\begin{equation}\label{eq:doubling-capacity}
 nmq+1\le2^{2n},\qquad m|\mathcal H|\le2^n.
\end{equation}
The first condition bounds the inputs and anchor in each copy; the second bounds its $m2^n|\mathcal H|$ candidate centers. The next lemma converts these counts into a bound on the full metric's doubling dimension.

\begin{lemma}\label{lem:doubling-copy}
The prescribed distances define a metric. Under \eqref{eq:doubling-capacity}, its doubling dimension satisfies
\[
 n+1\le\ddim(\calX,d)\le2n+1,
\]
independently of the number of copies and padding points.
\end{lemma}
\begin{proof}
Within a single copy, the inputs and anchor have pairwise distances in $\{1,2\}$, and the candidate centers have pairwise distance $2$. For triangles involving both inputs/anchor and candidate centers: a triangle with an input, the anchor, and a candidate center has side lengths $(1,R,R\pm1)$; a triangle with the anchor and two candidate centers has side lengths $(R,R,2)$; and a triangle with either two inputs and one center or one input and two centers has one side of length $2$ and two sides in $[R-1,R+1]$. Since $R\ge4$, all intra-copy triangle inequalities hold.

Across components, each component has diameter at most $R+1<5R$. A triangle spanning three distinct components reduces to the line metric in \eqref{eq:doubling-line-join}, while a triangle with two points in the same component and one in another has two equal sides of length at least $10R$ and a third side of length at most $R+1$. Thus $(\calX,d)$ is a metric space.

By \eqref{eq:doubling-capacity}, each component has size at most $2^{2n+1}$. Any ball of radius $r<T:=10R$ is contained in a single component and can therefore be covered by at most $2^{2n+1}$ balls of radius $r/2$ (one per point). For $r\ge T$, a ball of radius $r$ is the union of all components whose line indices $t$ lie in an interval of length at most $2\lfloor r/T\rfloor+1$. Partitioning this index interval into at most three sub-intervals of length at most $2\lfloor r/(2T)\rfloor+1$ and placing a ball of radius $r/2$ at the anchor of each sub-interval's middle component covers every component in the sub-interval: points in the middle component lie within its diameter $R+1<T/2\le r/2$, and points in the other components of the sub-interval lie at distance at most $\lfloor r/(2T)\rfloor T\le r/2$. Since $3\le 2^{2n+1}$ for $n\ge1$, this establishes $\ddim(\calX,d)\le 2n+1$.

For the lower bound, a radius-two ball centered at any candidate center contains all $m2^n|\mathcal H|$ candidate centers in its copy. No radius-one ball covers two of them: their pairwise distance is two, every input and anchor is at distance at least $R-1>1$ from them, and all other components are farther away. Thus the doubling constant is at least $m2^n|\mathcal H|\ge2^{n+1}$, using $m\ge1$ and $|\mathcal H|\ge2$.
\end{proof}

\subsection{Query center sets and costs}

Fix an integer $g$ with $1\le g\le B$ and a weighted subset $(S,w)$ with $|S|\le|X|/2$. We choose the $g$ copies containing the most points omitted by $S$ as active, breaking ties arbitrarily. Without loss of generality, denote them by $P_1,\ldots,P_g$. The line positions remain fixed when the active copies are relabeled. We select $m$ candidate centers in each active copy, one per label, and the anchor of each inactive copy.

For $\chi\in\mathcal H$, choose one vector $\sigma_{bi}$ for each active copy $b\in[g]$ and label $i\in[m]$, and write
\[
 \sigma=(\sigma_{bi})_{b\in[g],\,i\in[m]}
 \in\bigl(\{-1,1\}^n\bigr)^{g\times m},
 \qquad \sigma_{bi}=(\sigma_{bi}^{(j)})_{j\in[n]}.
\]
Thus $\sigma$ consists of $gm$ $n$-dimensional vectors, one for each selected active center; $j$ indexes the coordinates of each vector. Together with the inactive anchors, this uses $gm+(B-g)=B+g(m-1)$ centers. To stay within the center budget, we will choose the parameters so that
\begin{equation}\label{eq:doubling-center-budget}
 B+g(m-1)\le k.
\end{equation}
Adding fixed padding centers gives the following set of $k$ centers:
\begin{equation*}
\begin{aligned}
 C(\sigma,\chi)
 ={}&\{c_{b,i,\sigma_{bi},\chi}:b\in[g],\ i\in[m]\}\\
 &{}\cup\{a_b:g<b\le B\}\\
 &{}\cup\{a_{B+t}:1\le t\le k-B-g(m-1)\}.
\end{aligned}
\end{equation*}
With these active copies fixed, define
\begin{equation*}
 \mathcal C=
 \left\{C(\sigma,\chi):
 \sigma\in\bigl(\{-1,1\}^n\bigr)^{g\times m},\quad
 \chi\in\mathcal H\right\}.
\end{equation*}
In our error analysis, we pair each $C(\sigma,\chi)\in\mathcal C$ with $C(-\sigma,\chi)$, where $-\sigma$ negates all $gm$ sign vectors while keeping $\chi$, the active copies, and the padding centers fixed.

\begin{claim}[Pointwise and total costs]
For every $C(\sigma,\chi)\in\mathcal C$, the following properties hold.
\begin{itemize}
\item \emph{Pointwise costs.} In an active copy, each point satisfies
\begin{equation}\label{eq:doubling-pointwise-cost}
 d^z(p_{b,j,i,s},C(\sigma,\chi))
 =(R-\sigma_{bi}^{(j)}\chi_s)^z.
\end{equation}
Every point in an inactive copy has cost one.
\item \emph{Total cost.} The full dataset satisfies
\begin{equation}\label{eq:doubling-cost-bound}
 \cost_z(X,C(\sigma,\chi))\lesssim gR^znmq
\end{equation}
provided $B\le2gR^z$.
\end{itemize}
\end{claim}
\begin{proof}
In an active copy, the anchor is not selected. The point $p_{b,j,i,s}$ has distance $R-\sigma_{bi}^{(j)}\chi_s$ to its designated center and distance $R+1$ to each selected center of another label. All selected centers outside this copy have distance at least $10R>R+1$. Thus the designated center is always a nearest center, with possible ties when its distance is $R+1$, proving \eqref{eq:doubling-pointwise-cost}. An inactive point is served by its own anchor at distance one.

Summing these pointwise costs gives
\[
 \cost_z(X,C(\sigma,\chi))
 \le[g(R+1)^z+B-g]nmq.
\]
To obtain \eqref{eq:doubling-cost-bound}, it is sufficient to keep the inactive contribution within a constant factor of the active contribution. We therefore impose
\begin{equation}\label{eq:doubling-cost-balance}
 B\le2gR^z,
\end{equation}
which will be enforced in the parameter choice. Since $R\ge4$, the preceding cost estimate is then $O(gR^znmq)$.
\end{proof}

\subsection{Coreset error and parameter choice}

\begin{proposition}[Two center sets detecting omissions]\label{prop:doubling-discrepancy}
Under \eqref{eq:doubling-center-budget}, \eqref{eq:doubling-cost-balance}, and \eqref{eq:doubling-capacity}, for every weighted subset $(S,w)$ with $|S|\le|X|/2$, the family $\mathcal C$ constructed above for $S$ contains two center sets $C^+,C^-$, each of size exactly $k$, satisfying
\begin{equation}\label{eq:doubling-relative-error}
 \max_{C\in\{C^+,C^-\}}
 \frac{|\cost_z(S,C)-\cost_z(X,C)|}{\cost_z(X,C)}
 \gtrsim\frac1{R\sqrt q}.
\end{equation}
\end{proposition}
\begin{proof}
Fix $(S,w)$ and let $P_1,\ldots,P_g$ be the active copies selected above. Throughout this proof, sums over $b,j,i$ range over $[g]\times[n]\times[m]$, and sums over $s$ range over $[q]$. Extend $w$ by zero to omitted points, and set
\[
 A_{bji}=(w(p_{b,j,i,s})-1)_{s\in[q]}.
\]
Write the signed error as
\[
 \Delta(C)=\cost_z(S,C)-\cost_z(X,C)
          =\sum_{p\in X}(w(p)-1)d^z(p,C).
\]
For arbitrary $\chi\in\mathcal H$ and $\sigma\in\bigl(\{-1,1\}^n\bigr)^{g\times m}$, let $C^+=C(\sigma,\chi)$ and $C^-=C(-\sigma,\chi)$. For an active point, \eqref{eq:doubling-pointwise-cost} gives
\[
d^z(p_{b,j,i,s},C^-)-d^z(p_{b,j,i,s},C^+)
=\sigma_{bi}^{(j)}\chi_s\bigl[(R+1)^z-(R-1)^z\bigr]\asymp R^{z-1}\sigma_{bi}^{(j)}\chi_s.
\]
Every inactive point has the same cost in both center sets. Subtracting the errors therefore cancels all inactive contributions and yields
\begin{align*}
 \Delta(C^-)-\Delta(C^+)
 &=\sum_{p\in X}(w(p)-1)
      \bigl[d^z(p,C^-)-d^z(p,C^+)\bigr]\\
 &\asymp R^{z-1}
      \sum_{b,j,i}\sigma_{bi}^{(j)}\langle A_{bji},\chi\rangle.
\end{align*}
For fixed $\chi$, choose $\sigma_{bi}^{(j)}=\mathrm{sgn}(\langle A_{bji},\chi\rangle)$, with either choice when this is zero. The last sum then equals $\sum_{b,j,i}|\langle A_{bji},\chi\rangle|$. Applying Lemma~\ref{lem:doubling-signs} and linearity of expectation now gives
\begin{equation}\label{eq:doubling-response}
 \max_{\chi\in\mathcal H}\sum_{b,j,i}|\langle A_{bji},\chi\rangle|
 \ge\E_{\chi}\sum_{b,j,i}|\langle A_{bji},\chi\rangle|
 \gtrsim\sum_{b,j,i}\|A_{bji}\|_2.
\end{equation}
The same $\chi$ is used in every term of the sum.

Let $M_{bji}=|\{s\in[q]:p_{b,j,i,s}\notin S\}|$. Each omitted point has coefficient $-1$, so $\|A_{bji}\|_2\ge\sqrt{M_{bji}}$. The active copies contain the most omissions, hence their average number of omissions is at least the average over all $B$ copies. Since $|X\setminus S|\ge Bnmq/2$, we have $\sum_{b,j,i}M_{bji}\gtrsim gnmq$. Together with $0\le M_{bji}\le q$, this gives
\begin{equation}\label{eq:doubling-omissions}
 \sum_{b,j,i}\|A_{bji}\|_2
 \ge\sum_{b,j,i}\sqrt{M_{bji}}
 \ge\frac{\sum_{b,j,i}M_{bji}}{\sqrt q}
 \gtrsim gnm\sqrt q.
\end{equation}
Choose $\chi$ attaining the maximum in \eqref{eq:doubling-response} and align $\sigma$ as above. Combining \eqref{eq:doubling-response} and \eqref{eq:doubling-omissions} with $(R+1)^z-(R-1)^z\gtrsim R^{z-1}$ gives
\begin{equation*}
 \Delta(C^-)-\Delta(C^+)\gtrsim gR^{z-1}nm\sqrt q.
\end{equation*}
By the triangle inequality and \eqref{eq:doubling-cost-bound},
\begin{align*}
 \max_{C\in\{C^+,C^-\}}\frac{|\Delta(C)|}{\cost_z(X,C)}
 &\ge\frac{\Delta(C^-)-\Delta(C^+)}
              {\cost_z(X,C^-)+\cost_z(X,C^+)}\\
 &\gtrsim\frac{gR^{z-1}nm\sqrt q}{gR^znmq}
 =\frac1{R\sqrt q}.
\end{align*}
This proves \eqref{eq:doubling-relative-error}.
\end{proof}

\paragraph{Choosing the parameters.}
Set $n=\lfloor D/8\rfloor$. Fix a sufficiently small constant $0<\gamma_z\le1$ so that Proposition~\ref{prop:doubling-discrepancy} gives relative error greater than $\eps$ whenever
\[
\eps R\sqrt q\le\sqrt{\gamma_z}.
\]
Ignoring integer rounding for the moment, take $B\asymp_z k$ and $gm\asymp k$. For fixed $R$, the error condition allows $q\asymp_z\eps^{-2}R^{-2}$, while the cost-balance condition requires $g\gtrsim_z k/R^z$. Taking $g$ at this scale leaves the largest possible order of $m$ within the center budget, giving
\[
g\asymp_z k/R^z,\qquad
m\asymp_z R^z,\qquad
|X|\asymp_z kD\eps^{-2}R^{z-2}.
\]
The requirements $g\ge1$ and $q\ge1$ limit $R^z$ to order $\min\{k,\eps^{-z}\}$. Since $z\ge2$, the dataset size is nondecreasing in $R$. We therefore take
\[
k_0:=\min\{k,\lfloor\eps^{-z}\rfloor\},
\qquad R^z\asymp_z k_0.
\]
These scales yield the desired size:
\[
|X|\asymp_z kD\eps^{-2}k_0^{1-2/z}
\asymp_z kD\min\{k^{1-2/z}\eps^{-2},\,\eps^{-z}\}.
\]

To see why we allow $g$ active copies rather than fixing $g=1$, consider the parameter scales established above. The discrepancy argument gives
\[
 \text{error gap}=\Omega(g n m^{2-1/z}\sqrt q),\qquad
 \cost_z(X,C)=\Theta((gm+B)nmq).
\]
For $g=1$, $B\asymp k$, and $m\le k$, the relative-error lower bound certified by this argument has scale
\[
 \frac{m^{1-1/z}}{(m+B)\sqrt q}
 \lesssim\frac{k^{-1/z}}{\sqrt q}.
\]
Making this guarantee exceed $\eps$ would require $q\lesssim\eps^{-2}k^{-2/z}$, which is incompatible with $q\ge1$ when $k/\eps^{-z}$ is unbounded. Thus this discrepancy argument with a single active copy does not yield the required dataset size throughout the parameter range. Allowing $g\asymp k/m$ lets us keep $m\lesssim\eps^{-z}$ while using $\Theta(k)$ active centers; these parameter scales give the claimed lower bound.

\paragraph{Integral choices and dimension.}
Choose $\rho_z>0$ with
\[
\rho_z\le\min\{1/4,\sqrt{\gamma_z}/2\},
\qquad \rho_z^z\le1/8,
\]
and choose $\eps_z\le1/2$ small enough that $\gamma_z\eps_z^{-2}/16\ge4$. For $k_0$ above a sufficiently large constant depending only on $z$, set
\[
\begin{gathered}
g=\lfloor k/k_0\rfloor,\qquad
R=\rho_z k_0^{1/z},\qquad
B=g\lceil R^z\rceil,\\
m=\lfloor k_0/2\rfloor,\qquad
q=\lfloor\gamma_z\eps^{-2}R^{-2}\rfloor.
\end{gathered}
\]
The threshold ensures $R\ge4$ and $m\ge1$. Since $k_0\le\eps^{-z}$, the unrounded value defining $q$ is at least $\gamma_z/\rho_z^2\ge4$. Thus $q\ge1$, $q\asymp_z\eps^{-2}R^{-2}$, and $\eps R\sqrt q\le\sqrt{\gamma_z}$. Also $1\le g\le B\le2gR^z$, and
\[
B+g(m-1)
=g(\lceil R^z\rceil+m-1)
\le g(R^z+k_0/2)
\le gk_0
\le k.
\]
Hence the center budget and cost balance hold. Since $gk_0\asymp k$, these integral choices preserve the parameter scales above.

Taking $C_z$ sufficiently large in the dimension assumption ensures
\begin{equation}\label{eq:doubling-dimension-slack}
n\ge16(z+20)\log_2(2/\eps).
\end{equation}
Using $m\le\eps^{-z}$, $q\le\eps^{-2}$, and Lemma~\ref{lem:doubling-signs}, we obtain
\[
\begin{aligned}
\log_2(m|\mathcal H|)
&\le(z+14)\log_2(2/\eps)\le n,\\
\log_2(nmq+1)
&\le\log_2 n+(z+4)\log_2(2/\eps)\le2n.
\end{aligned}
\]
Here the last inequality uses $\log_2 n\le n$ and \eqref{eq:doubling-dimension-slack}. These estimates verify \eqref{eq:doubling-capacity}, so Lemma~\ref{lem:doubling-copy} gives
\[
D/8<n+1\le\ddim(\calX,d)\le2n+1\le D.
\]
The capacity estimates depend on $\eps$ and $z$, but not on the number of copies, so no $\log k$ term is needed in the dimension assumption. Together with Proposition~\ref{prop:doubling-discrepancy} and the size calculation above, this proves the claim for sufficiently large $k_0$.

For bounded $k_0$, keep $n=\lfloor D/8\rfloor$ and use
\[
B=g=k,\qquad m=1,\qquad R=4,\qquad
q=\lfloor\gamma_z\eps^{-2}/16\rfloor.
\]
The choice of $\eps_z$ ensures $q\ge1$. The center budget, cost balance, and error condition hold, and the same capacity estimates apply. This gives $|X|\asymp_z kD\eps^{-2}$, which has the required order because $k_0^{1-2/z}=\Theta_z(1)$ in this case. Proposition~\ref{prop:doubling-discrepancy} completes the proof of Lemma~\ref{lem:doubling-lower}.
\begin{remark}[$k$-means]
For $z=2$, the constant-scale choice $B=g=k$, $m=1$, $R=4$, and $q=\lfloor\gamma_2\eps^{-2}/16\rfloor$ already gives $|X|\asymp kD\eps^{-2}$ for every $k$.
\end{remark}

\subsection{Proof of the sign-family lemma}\label{sec:doubling-sign-proof}
To prove Lemma~\ref{lem:doubling-signs}, we use the binary signs from the Euclidean code in Lemma~\ref{lem:euclidean-code}, without the projection levels.

\begin{proof}[Proof of Lemma~\ref{lem:doubling-signs}]
Let $L\ge1$ be the least integer with $2^L>q+\binom q3$. Use the distinct binary words $b_1,\ldots,b_q\in\{0,1\}^L$ constructed in the proof of Lemma~\ref{lem:euclidean-code}, no four distinct of which sum to zero coordinatewise modulo two. For independent uniform signs $r_{a,\beta}$, $a\in[L]$, $\beta\in\{0,1\}$, set
\[
 \chi_s=\prod_{a=1}^L r_{a,b_s(a)}.
\]
The parity-moment calculation from the Euclidean code proof gives unbiased four-wise independent coordinates. There are at most $2^{2L}$ outcomes, and minimality of $L$ gives $2^L\le2(q+\binom q3+1)\le2(q+2)^3$. This proves the upper cardinality bound; unbiasedness gives $|\mathcal H|\ge2$. For $Z=\sum_sA_s\chi_s$, the second and fourth moments satisfy
\[
 \E Z^2=\sum_sA_s^2,\qquad
 \E Z^4\le3\left(\sum_sA_s^2\right)^2.
\]
If $A\ne0$, H\"older's inequality gives $\E|Z|\ge(\E Z^2)^{3/2}/(\E Z^4)^{1/2}\ge\|A\|_2/\sqrt3$. The case $A=0$ is immediate.
\end{proof}

\section{Upper bound in doubling metrics}
\label{sec:doubling-upper-proof}

The following lemma supplies the upper bound used in Theorem~\ref{thm:intro-doubling}. Throughout this section, $X$ has input weights $u_p>0$, and $\cost_z(X,C)=\sum_{p\in X}u_pd^z(p,C)$; input subsets inherit these weights. The output $S$ has weights $w_p$, with $\cost_z(S,C)=\sum_{p\in S}w_pd^z(p,C)$. All cardinalities count support points.

\begin{lemma}[Coreset upper bound in doubling metrics]\label{lem:doubling-upper}
For every fixed real $z\ge2$, integer $k\ge1$, $D\ge1$, and $0<\eps<1$, every finite positively weighted dataset $X$ in a metric space $(\calX,d)$ with $\ddim(\calX,d)\le D$ admits a nonnegatively weighted subset $S\subseteq X$ of size
\begin{equation}\label{eq:doubling-upper}
 O\!\left(k^{2-2/z}\eps^{-2}
       [D\log(2/\eps)+\log(2k)]\right)
\end{equation}
such that
\[
 |\cost_z(S,C)-\cost_z(X,C)|\le\eps\cost_z(X,C)
 \qquad\text{for every }C\subseteq\calX\text{ with }1\le|C|\le k.
\]
\end{lemma}

\subsection{Coreset construction and proof of Lemma \ref{lem:doubling-upper}}

Let $n=|X|$. Choose a subset $A\subseteq X$ with $1\le|A|\le k$ minimizing $\cost_z(X,A)$; then $\cost_z(X,A)\le2^z\cost_z(X,C)$ for every query $C$ with $1\le|C|\le k$. Write $A=\{a_1,\ldots,a_{|A|}\}$ and partition $X$ into $P_1,\ldots,P_{|A|}$ by nearest-anchor assignment, breaking ties arbitrarily. We use the optimum only to simplify constants; approximate anchors are discussed at the end of this section.
Within each cluster $P_i$, the algorithm draws $s_i$ independent two-point weighted summaries and averages their weight vectors:

\newcounter{doublingalgorithm}
\newenvironment{doublingalgorithm}[1]{
  \par\addvspace{1em}\noindent
  \begin{minipage}{\linewidth}
  \refstepcounter{doublingalgorithm}
  
  \hrule\smallskip
  \textbf{Algorithm \thedoublingalgorithm. #1}\par\smallskip\small
}{
  \par\smallskip\hrule
  
  \end{minipage}\par\addvspace{1em}
}

\begin{doublingalgorithm}{Coreset construction via two-point sampling}
\label{alg:doubling-stage}
\textbf{Input:} $X$ with weights $u_p>0$, anchors $A$ and clusters $P_i$ as above, and $0<\eps,\beta<1$.
\begin{enumerate}
\setlength{\itemsep}{2pt}\setlength{\parskip}{0pt}
\item If $\cost_z(X,A)=0$, return $X$ with its input weights.
\item For each $i$, set
\begin{equation}\label{eq:doubling-allocation}
 s_i=1+\left\lceil K_z\eps^{-2}
       \left(\frac{\cost_z(P_i,A)}{\cost_z(X,A)}\right)^{2/z}
       [k\log(2n)+kD\log(2/\eps)+\log(2/\beta)]
     \right\rceil,
\end{equation}
where $K_z>0$ is a sufficiently large constant depending only on $z$.
\item Independently for each $i$ and $j\in[s_i]$, draw $(\widehat u_p^{(j)})_{p\in P_i}$ using the distribution in Lemma~\ref{lem:doubling-quadrature} with weights $(u_p)_{p\in P_i}$ and $h_p=d^z(p,A)$. Set
\begin{equation}\label{eq:doubling-average}
 w_p=\frac1{s_i}\sum_{j=1}^{s_i}\widehat u_p^{(j)} \qquad(p\in P_i).
\end{equation}
\item Return $(S,(w_p)_{p\in S})$, where $S=\{p\in X:w_p>0\}$.
\end{enumerate}
\end{doublingalgorithm}

The next lemma supplies the two-point sampling distribution for Step~3.

\begin{lemma}[Two-point unbiased sampling]\label{lem:doubling-quadrature}
Let $(u_p)$ be nonnegative weights on a finite set and let $h_p\ge0$. There is a finite distribution over nonnegative weight vectors $(\widehat u_p)$, each supported on at most two points of $\{p:u_p>0\}$, such that
\[
 \sum_p\widehat u_p=\sum_pu_p,\qquad
 \sum_p\widehat u_ph_p=\sum_pu_ph_p,\qquad
 \E\widehat u_p=u_p\quad\text{for every }p.
\]
\end{lemma}
\begin{proof}
Consider the bounded polytope of nonnegative weight vectors $(\widehat u_p)$ supported on $\{p:u_p>0\}$ that satisfy the two linear equations. There are two equality constraints and nonnegativity constraints, so every vertex (basic feasible solution) has at most two positive entries. Since $(u_p)$ lies in this polytope, we can decompose $(u_p)$ as a convex combination of its vertices and sample a vertex with probability equal to its convex coefficient.
\end{proof}

Each draw in Step~3 is supported on at most two points of $P_i$ and preserves the cluster's total weight and anchor cost. Consequently, the averaged weights $(w_p)_{p\in X}$ satisfy, deterministically for every $i$,
\begin{equation}\label{eq:doubling-sample-totals}
\sum_{p\in P_i}w_p=\sum_{p\in P_i}u_p,\qquad
\sum_{p\in P_i}w_p d^z(p,A)=\sum_{p\in P_i}u_p d^z(p,A)=\cost_z(P_i,A).
\end{equation}
The averaged weights are also unbiased ($\E w_p=u_p$ for all $p\in X$). Together, these properties give the following guarantee for Algorithm~\ref{alg:doubling-stage}.

\begin{lemma}[Coreset guarantee for Algorithm~\ref{alg:doubling-stage}]\label{lem:doubling-finite}
For every $0<\eps,\beta<1$, Algorithm~\ref{alg:doubling-stage} returns a nonnegatively weighted subset $S\subseteq X$ of size at most
\begin{equation}\label{eq:doubling-finite}
 O\!\left(k^{1-2/z}\eps^{-2}
       [k\log(2n)+kD\log(2/\eps)+\log(2/\beta)]\right)
\end{equation}
that is a strong $\eps$-coreset of $X$ with probability at least $1-\beta$.
\end{lemma}

To see the size bound in Lemma~\ref{lem:doubling-finite}, note first that if $\cost_z(X,A)=0$, then $|X|\le k$ and the bound holds trivially. Otherwise, each of the $s_i$ draws in cluster $P_i$ contributes at most two points to $S$, so
\begin{equation}\label{eq:doubling-sample-count}
\begin{aligned}
|S|\le 2\sum_i s_i&\lesssim \eps^{-2}\left[k\log(2n)+kD\log(2/\eps)+\log(2/\beta)\right]
\sum_i\left(\frac{\cost_z(P_i,A)}{\cost_z(X,A)}\right)^{2/z}\\
&\lesssim k^{1-2/z}\eps^{-2}\left[k\log(2n)+kD\log(2/\eps)+\log(2/\beta)\right].
\end{aligned}
\end{equation}
The last inequality follows from the concavity of $t^{2/z}$ for $z\ge2$ and $\sum_i\cost_z(P_i,A)=\cost_z(X,A)$. This proves the size bound in Lemma~\ref{lem:doubling-finite}; we defer its approximation guarantee to the next subsection.

\paragraph{Removing the input-size dependence.}
We first reduce to polynomial input size, then apply the sharper sampling bound once more. For $0<\eta,\beta<1/4$, Lemma~\ref{lem:doubling-finite} applies to weighted inputs and yields a coreset of size $O(k^{2-2/z}(D+1)\eta^{-3}\log(2n)\log(2/\beta))$, using $\log(2/\eta)=O(\eta^{-1})$. Applying Theorem~\ref{thm:iterative-size-reduction} with $\rho=3$ and $s(k)=k^{2-2/z}(D+1)$ produces an $\eps/4$-coreset $X'\subseteq X$ of size polynomial in $k,D,\eps^{-1}$, with probability at least $0.99$. In particular,
\[
 \log(2|X'|)=O\!\left(\log(2k)+\log(D+1)+\log(2/\eps)\right).
\]
Applying Lemma~\ref{lem:doubling-finite} once more to $X'$ with its coreset weights, accuracy $\eps/2$, and failure probability $1/100$, gives
\begin{align*}
 |S|
 &\lesssim k^{1-2/z}\eps^{-2}\left[k\log(2|X'|)+kD\log(2/\eps)+1\right]\\
 &\lesssim k^{2-2/z}\eps^{-2}\left[D\log(2/\eps)+\log(2k)\right],
\end{align*}
where $\log(D+1)=O(D)$ for $D\ge1$. With constant probability, both stages succeed; coreset composition then makes $S\subseteq X'\subseteq X$ an $\eps$-coreset, proving Lemma~\ref{lem:doubling-upper}.

\paragraph{Matching bounds.}
Combining \eqref{eq:doubling-upper} with the known $\widetilde O(kD\eps^{-z})$ upper bound~\cite{cohenaddad2021newframework,DBLP:conf/stoc/Huang0024} gives the upper bound in Theorem~\ref{thm:intro-doubling}.

\subsection{Proof of Lemma~\ref{lem:doubling-finite}: Error analysis}

It remains to prove the approximation guarantee of Lemma~\ref{lem:doubling-finite}. The case $\cost_z(X,A)=0$ is exact, so assume $\cost_z(X,A)>0$. Fix a query $C$. Clusters of zero anchor cost contain only their anchor and have zero residual; omit them from the range calculations. Exact cluster-mass preservation gives
\[
 \cost_z(S,C)-\cost_z(X,C)=\sum_i\sum_{p\in P_i}(w_p-u_p)\bigl(d^z(p,C)-d^z(a_i,C)\bigr).
\]

To bound the fluctuations, we truncate the cost difference when $C$ is far from $a_i$ relative to $d(p,A)$ (where $d^z(p,C)\approx d^z(a_i,C)$ deterministically) and define, for $p\in P_i$,
\begin{equation}\label{eq:doubling-truncated}
 g_C(p)=\begin{cases}
 d^z(p,C)-d^z(a_i,C),&d(p,A)\ge 2^{-3z}\eps\,d(a_i,C),\\
 0,&\text{otherwise}.
 \end{cases}
\end{equation}
By \eqref{eq:doubling-average},
\[
 \sum_{p\in X}(w_p-u_p)g_C(p)
 =\sum_i\frac1{s_i}\sum_{j=1}^{s_i}\sum_{p\in P_i}(\widehat u_p^{(j)}-u_p)g_C(p).
\]
The summands indexed by $(i,j)$ are independent by construction and centered by unbiasedness in Lemma~\ref{lem:doubling-quadrature}: independence is between weighted draws, not between coordinates within a draw. We bound the sum of their squared ranges to apply Hoeffding's inequality.

\paragraph{Bounding the squared ranges.}
For $p\in P_i$, Lemma~\ref{lem:triangle-inequalities} gives
\begin{equation}\label{eq:doubling-residual-bound}
 |g_C(p)|\le |d^z(p,C)-d^z(a_i,C)|
 \lesssim d^z(p,A)+d(p,A)d(a_i,C)^{z-1}.
\end{equation}
For the fixed query $C$, let $R_i$ be the range length of $s_i^{-1}\sum_{p\in P_i}\widehat u_p^{(j)}g_C(p)$ over all draws. This range is the same for every $j\in[s_i]$ and is unchanged by centering the summand. By H\"older's inequality (Lemma~\ref{lem:holder-inequality}) and the two preserved totals in \eqref{eq:doubling-sample-totals},
\[
 \sum_{p\in P_i}\widehat u_p^{(j)}d(p,A)
 \le\left(\sum_{p\in P_i}\widehat u_p^{(j)}\right)^{1-1/z}
     \left(\sum_{p\in P_i}\widehat u_p^{(j)}d^z(p,A)\right)^{1/z}
 =\left(\sum_{p\in P_i}u_p\right)^{1-1/z}\cost_z(P_i,A)^{1/z}.
\]
Combining the preceding moment bound with \eqref{eq:doubling-residual-bound} controls each uncentered draw and hence its range length:
\begin{equation}\label{eq:doubling-sample-range}
 R_i\lesssim\frac1{s_i}\left[
 \cost_z(P_i,A)+d(a_i,C)^{z-1}\left(\sum_{p\in P_i}u_p\right)^{1-1/z}
                      \cost_z(P_i,A)^{1/z}\right].
\end{equation}
There are $s_i$ independent terms of range length $R_i$ in cluster $P_i$. Squaring \eqref{eq:doubling-sample-range}, summing over clusters, and substituting the choice of $s_{i}$ from \eqref{eq:doubling-allocation} yields
\begin{align*}
 \sum_i s_iR_i^2
 &\lesssim\sum_i
 \frac{\cost_z(P_i,A)^2+\left(\sum_{p\in P_i}u_pd^z(a_i,C)\right)^{2-2/z}\cost_z(P_i,A)^{2/z}}{s_i}\\
 &\lesssim
 \frac{\eps^2\cost_z(X,A)^{2/z}}{K_z[k\log(2n)+kD\log(2/\eps)+\log(2/\beta)]}\\
 &\qquad{}\cdot\sum_i\left[\cost_z(P_i,A)^{2-2/z}+\left(\sum_{p\in P_i}u_pd^z(a_i,C)\right)^{2-2/z}\right].
\end{align*}
Since $2-2/z\ge1$, each sum of nonnegative powers is at most the corresponding power of the sum. Combining this observation with Lemma~\ref{lem:triangle-inequalities} gives
\begin{align*} \sum_i\sum_{p\in P_i}u_pd^z(a_i,C)&\le2^{z-1}[\cost_z(X,C)+\cost_z(X,A)],\\ \sum_i\cost_z(P_i,A)^{2-2/z}&\le\cost_z(X,A)^{2-2/z},\\ \sum_i\left(\sum_{p\in P_i}u_pd^z(a_i,C)\right)^{2-2/z}&\lesssim[\cost_z(X,C)+\cost_z(X,A)]^{2-2/z}. \end{align*}
Substituting these estimates yields
\begin{equation} \sum_i s_iR_i^2\lesssim\frac{\eps^2[\cost_z(X,C)+\cost_z(X,A)]^2}{K_z[k\log(2n)+kD\log(2/\eps)+\log(2/\beta)]}. \end{equation}

Hoeffding's inequality (Lemma~\ref{lem:hoeffding}) therefore gives
\begin{equation}\label{eq:doubling-concentration-scale}
 \left|\sum_{p\in X}(w_p-u_p)g_C(p)\right|
 \le2^{-z-6}\eps[\cost_z(X,C)+\cost_z(X,A)]
\end{equation}
except with probability at most
\begin{equation}\label{eq:doubling-fixed-query-tail}
 2\exp\!\left(-c_zK_z
       [k\log(2n)+kD\log(2/\eps)+\log(2/\beta)]\right),
\end{equation}
where $c_z>0$ depends only on $z$. 

The preceding concentration bound holds for each fixed $C$. To obtain a guarantee for all queries $C$ simultaneously, we use the following net.

\begin{lemma}[Covering the truncated cost differences]\label{lem:doubling-covering}
There is a finite family $V\subseteq\R^X$ and a map associating each query $C$ with a vector $v\in V$, fixed before sampling, such that
\begin{align}
 |g_C(p)-v_p|&\le2^{-z-6}\eps d^z(p,A)\qquad(p\in X),\label{eq:doubling-vector-error}\\
 \log|V|&\lesssim k\log(2n)+kD\log(2/\eps).\label{eq:doubling-entropy}
\end{align}
\end{lemma}
\begin{proof}
Write $r_p=d(p,A)$ and $L=2^{3z}\eps^{-1}+2$. Points with $r_p=0$
equal their anchors and have $g_C(p)=0$ for every $C$; if all scales
are zero, take $V=\{0\}$. Otherwise, fix
$\delta=2^{-z-8}\eps/(zL^{z-1})\asymp\eps^z$.
For each $p\in P_i$ with $r_p>0$, fix a cover of $B(a_i,Lr_p)$
by balls of radius $\delta r_p$. Give every covering ball its own
index, retaining its radius even if its center occurs in other covers.
For each $c\in C$, record the index of a smallest-radius covering
ball containing $c$, with ties broken by a fixed order; use a null
record if no ball contains $c$. Sort the records and pad to $k$ slots
with null records. Together with the active suffix in each cluster,
as defined by \eqref{eq:doubling-truncated}, these slots form the query's
code. All covers and ordering rules are fixed before sampling.

Consider queries $C,C'$ with the same code and an active
$p\in P_i$ with $r_p>0$. At this coordinate, keep only records
of radius at most $\delta r_p$.
For $x\in\{p,a_i\}$, a center $c\in C$ nearest to $x$ satisfies
$d(a_i,c)\le d(a_i,C)+2r_p\le Lr_p$.
It therefore has a record of radius at most $\delta r_p$ and is
not discarded. The same record contains a center $c'\in C'$, so
$d(c,c')\le2\delta r_p$ and
$d(x,C')\le d(x,C)+2\delta r_p$.
The active suffixes agree, so exchanging $C,C'$ gives
$|d(x,C)-d(x,C')|\le2\delta r_p$.
Thus discarding coarser and null records changes neither nearest distance.
Both distances are at most $Lr_p$, by the cutoff and the triangle
inequality. The mean value theorem now gives
\begin{equation*}
 |g_C(p)-g_{C'}(p)|\le4z\delta L^{z-1}r_p^z
 =2^{-z-6}\eps r_p^z.
\end{equation*}
Inactive coordinates and zero-scale coordinates agree exactly.

By doubling, the total number $N$ of indexed covering balls is at most
$n(O(L/\delta))^D$. Sorting each cluster by $r_p$ gives at most $n+1$
active suffixes, including the empty suffix; equal radii are kept or
discarded together. There are consequently at most
$(N+1)^k(n+1)^k$ codes. Choose one query $C_E$ for each attained code
$E$, set $V=\{g_{C_E}\}$, and map $C$ to $g_{C_{E(C)}}$.
These choices are fixed before sampling. The preceding comparison
proves \eqref{eq:doubling-vector-error}; counting the codes gives
\eqref{eq:doubling-entropy}, since $L/\delta=O_z(\eps^{-(z+1)})$.
\end{proof}

\paragraph{Uniform coreset guarantee.}
For each attained $v\in V$, fix before sampling a representative query $C_v$ whose cost $\cost_z(X,C_v)+\cost_z(X,A)$ is within a factor of two of the infimum over queries associated with $v$. Thus, for every $C$ associated with $v$,
\[
 \cost_z(X,C_v)+\cost_z(X,A)\le2[\cost_z(X,C)+\cost_z(X,A)],
\]
The infimum is positive because $\cost_z(X,A)>0$, so such a representative exists. Its concentration error is bounded in terms of the cost of every query it represents. There are at most $|V|$ such fixed representatives. By \eqref{eq:doubling-fixed-query-tail} and \eqref{eq:doubling-entropy}, a union bound shows that the probability that \eqref{eq:doubling-concentration-scale} fails for any of them is at most
\[
 2|V|\exp\!\left(-c_zK_z[k\log(2n)+kD\log(2/\eps)+\log(2/\beta)]\right)\le\beta,
\]
for sufficiently large $K_z$. Hence, with probability at least $1-\beta$, concentration holds for every $C_v$. On this event, combining the cost comparison above and $|g_C(p)-g_{C_v}(p)|\le 2^{-z-5}\eps d^z(p,A)$ with exact anchor-cost preservation and $\cost_z(X,A)\le 2^z\cost_z(X,C)$ gives, for every $C$ associated with $v$,
\begin{equation}
 \left|\sum_{p\in X}(w_p-u_p)g_C(p)\right|
 \le\left|\sum_{p\in X}(w_p-u_p)g_{C_v}(p)\right|+2^{-z-4}\eps\cost_z(X,A)
 \le\frac{\eps}{8}\cost_z(X,C).
\end{equation}
For points removed by the cutoff, $g_C(p)=0$ and $d(p,A)<2^{-3z}\eps\,d(a_i,C)$. The latter condition implies $d(a_i,C)\le 2d(p,C)$, so Lemma~\ref{lem:triangle-inequalities} gives $|d^z(p,C)-d^z(a_i,C)|\le\frac{\eps}{4}d^z(p,C)$. Combining the sampled and truncated contributions over $X$, with $|w_p-u_p|\le w_p+u_p$, yields
\[
 |\cost_z(S,C)-\cost_z(X,C)|
 \le\frac{\eps}{8}\cost_z(X,C)+\frac{\eps}{4}\bigl[\cost_z(X,C)+\cost_z(S,C)\bigr],
\]
which implies $|\cost_z(S,C)-\cost_z(X,C)|\le\frac{5\eps/8}{1-\eps/4}\cost_z(X,C)<\eps\cost_z(X,C)$, proving Lemma~\ref{lem:doubling-finite}.

\paragraph{Approximate anchors.}
Write $Q_C=\cost_z(X,C)$ and $V=\cost_z(X,A)$.
The argument also applies to any $A\subseteq X$ with $1\le|A|\le k$
and $V\le\alpha Q_C$ for every query, where $1\le\alpha=O_z(1)$,
using the nearest-anchor partition.
For target accuracy $\eps$ in a sampling stage, use internal accuracy
$\eta=\eps/[8(1+\alpha)]$ in place of $\eps$.
Before applying the anchor comparison, the preceding proof bounds
the error on truncated residuals by $2^{-z-6}\eta(2Q_C+6V)$ and the discarded
error by $\eta[Q_C+\cost_z(S,C)]/4$.
Substituting $V\le\alpha Q_C$ and moving the sampled-error term
to the left gives relative error at most
$\eta[2^{-z-6}(2+6\alpha)+1/2]/(1-\eta/4)<\eps$.
This rescales both the sampling and covering accuracies, changing
only constants depending on $z$ and $\alpha$.
At every iterative size-reduction stage, choose such anchors for
the current weighted input and rescale that stage's target accuracy
in the same way.

\section*{AI methodology}

We used ChatGPT (OpenAI) to help develop and draft parts of the proofs
and to revise the exposition. The authors revised and verified the arguments, checked the references,
and take full responsibility for the paper.

\bibliographystyle{alphaurl}
\bibliography{sample}

\appendix 

\section{Auxiliary results}

The following inequalities and size-reduction theorem supply the auxiliary tools used in Section~\ref{sec:doubling-upper-proof}.

\begin{lemma}[Metric inequalities]\label{lem:triangle-inequalities}
Let $(\calX,d)$ be a metric space, let $C\subseteq\calX$ be nonempty and finite, and let $z\ge2$ be real. For all $p,q\in\calX$,
\begin{align*}
 d^z(p,C)&\le2^{z-1}\bigl[d^z(p,q)+d^z(q,C)\bigr],\\
 |d^z(p,C)-d^z(q,C)|&\le z2^{z-2}\bigl[d^z(p,q)+d(p,q)d(q,C)^{z-1}\bigr].
\end{align*}
\end{lemma}
\begin{proof}
The triangle inequality gives $d(p,C)\le d(p,q)+d(q,C)$ and $|d(p,C)-d(q,C)|\le d(p,q)$. The first bound follows from $(x+y)^z\le2^{z-1}(x^z+y^z)$ for $x,y\ge0$.

For the second bound, the mean value theorem gives, for all $x,y\ge0$,
\begin{equation}
 |x^z-y^z|\le z|x-y|\max\{x,y\}^{z-1}.
\end{equation}
Substitute $x=d(p,C)$ and $y=d(q,C)$. The triangle inequality bounds $\max\{d(p,C),d(q,C)\}\le d(q,C)+d(p,q)$; applying $(x+y)^{z-1}\le2^{z-2}(x^{z-1}+y^{z-1})$ then gives the second bound.
\end{proof}

\begin{lemma}[H\"older's inequality]\label{lem:holder-inequality}
For $z>1$ and nonnegative finite sequences $(a_p)$ and $(b_p)$,
\[
 \sum_p a_pb_p\le
 \left(\sum_p a_p^{z/(z-1)}\right)^{1-1/z}
 \left(\sum_p b_p^z\right)^{1/z}.
\]
In particular, for nonnegative $(v_p)$ and $(t_p)$,
\[
 \sum_p v_pt_p\le
 \left(\sum_p v_p\right)^{1-1/z}
 \left(\sum_p v_pt_p^z\right)^{1/z}.
\]
\end{lemma}

\begin{lemma}[Hoeffding's inequality~{\cite{hoeffding1963probability}}]\label{lem:hoeffding}
Let $Y_1,\ldots,Y_m$ be independent real-valued random variables, and let $a_j\le b_j$ be finite deterministic bounds such that $a_j\le Y_j\le b_j$ almost surely for every $j\in[m]$. For every $t>0$,
\[
 \Pr\!\left[\left|\sum_{j=1}^m(Y_j-\E Y_j)\right|\ge t\right]
 \le2\exp\!\left(-\frac{2t^2}{\sum_{j=1}^m(b_j-a_j)^2}\right).
\]
If $\sum_{j=1}^m(b_j-a_j)^2=0$, the probability on the left is zero.
\end{lemma}
The two-sided bound follows by applying Hoeffding's one-sided inequality to $(Y_j)_{j=1}^m$ and $(-Y_j)_{j=1}^m$ and taking a union bound.

\begin{theorem}[Iterative size reduction~{\cite[Theorem~3.1]{braverman2021excludedminor}}]\label{thm:iterative-size-reduction}
Fix $z\ge1$, a constant $\rho\ge1$, and a family $\calM$ of metric spaces. Suppose that, for some function $s:\N\to[1,\infty)$, a randomized algorithm constructs, for every integer $k\ge1$, every finite positively weighted dataset $X\subseteq\calX$ with $(\calX,d)\in\calM$, and every $0<\eta,\beta<1/4$, a strong $\eta$-coreset for $(k,z)$-clustering supported on $X$ of size
\[
 O\!\left(\eta^{-\rho}s(k)\log(2|X|)\log(2/\beta)\right)
\]
with probability at least $1-\beta$. Then iterating this algorithm with suitable accuracies and failure probabilities constructs, for every $0<\eps,\beta<1/4$, a strong $\eps$-coreset supported on $X$ of size
\[
 \widetilde O\!\left(\eps^{-\rho}s(k)\log(2/\beta)\right)
\]
with probability at least $1-\beta$. Here $|X|$ counts support points, and $\widetilde O$ hides only logarithmic factors in $\eps^{-1}$, $s(k)$, and $\log(2/\beta)$.
\end{theorem}

\end{document}